\documentclass[times,doublespace]{oupau}
\usepackage[T1]{fontenc}
\usepackage[english]{babel}
\usepackage{bbm,enumerate}
\usepackage[linktocpage=true,colorlinks=true, linkcolor=blue, citecolor=red, urlcolor=green]{hyperref}
\usepackage[all]{xy}
\usepackage{color}

\newcommand{\mc}[1]{\mathcal{#1}}
\newcommand{\mf}[1]{\mathfrak{#1}}
\newcommand{\mb}[1]{\mathbb{#1}}

\newcommand{\mbb}[1]{\mathbbm{#1}}
\newcommand{\quot}[2] {\ensuremath{\raisebox{.40ex}{\ensuremath{#1}}
\! \big / \! \raisebox{-.40ex}{\ensuremath{#2}}}}

\newcommand{\tint}{{\textstyle\int}}
\newcommand{\ass}[1]{\stackrel{#1}{\longleftrightarrow}}

\DeclareMathOperator{\Mat}{Mat}

\DeclareMathOperator{\tr}{Tr}
\DeclareMathOperator{\res}{Res}
\DeclareMathOperator{\ord}{ord}

\DeclareMathOperator{\Span}{Span}

\newtheorem{remark}[theorem]{Remark}

\numberwithin{equation}{section}

\definecolor{light}{gray}{.9}

\begin{document}

\title{Adler-Gelfand-Dickey approach to classical
\texorpdfstring{$\mc W$}{W}-algebras within the theory of Poisson vertex algebras}
\shorttitle{AGD approach to classical W-algebras}

\volumeyear{2014}
\paperID{xxx}

\author{Alberto De Sole\affil{1}, Victor G. Kac\affil{2} and Daniele Valeri\affil{3}}
\abbrevauthor{A. De Sole, V. Kac and D. Valeri}
\headabbrevauthor{A. De Sole, V. Kac and D. Valeri}

\address{%
\affilnum{1}Dip.to Matematica, Universit\`a La Sapienza,
P.le Aldo Moro 2, 00185 Rome, Italy, 
desole@mat.uniroma1.it,
\affilnum{2}
Department of Mathematics, MIT,
77 Mass. Avenue, Cambridge, MA 02139, USA,
kac@math.mit.edu,
and
\affilnum{3}
SISSA, Via Bonomea 265, 34136 Trieste, Italy,
dvaleri@sissa.it.
}

\correspdetails{
desole@mat.uniroma1.it
}

\received{9 January 2014}
\revised{xxx}
\accepted{xxx}


\begin{abstract}
We put the Adler-Gelfand-Dickey approach to classical $\mc W$-algebras in the framework of
Poisson vertex algebras.
We show how to recover the bi-Poisson structure of the KP hierarchy,
together with its generalizations and reduction to the $N$-th KdV hierarchy,
using the formal distribution calculus and the $\lambda$-bracket formalism.
We apply the Lenard-Magri scheme
to prove integrability of the corresponding hierarchies.
We also give a simple proof of a theorem of Kupershmidt and Wilson in this framework.
Based on this approach,
we generalize all these results to the matrix case.
In particular, we find (non-local) bi-Poisson structures of the matrix KP 
and the matrix $N$-th KdV hierarchies, and 
we prove integrability of the $N$-th matrix KdV hierarchy.
\end{abstract}

\maketitle
\tableofcontents

\section{Introduction}
It is well known that classical $\mc W$-algebras play an important role
in the theory of integrable
bi-Hamiltonian equations.
One of the most powerful approaches to classical $\mc W$-algebras is via 
the Drinfeld-Sokolov Hamiltonian reduction, \cite{DS85},
associated to any simple Lie algebra and its principal nilpotent element.
Later on this theory was extended to more general nilpotent elements,
and  it was put by \cite{DSKV13a} in the framework of Poisson vertex algebras (PVA).

However the most important classical $\mc W$-algebras,
those associated to $\mf{sl}_N$ and its principal nilpotent element,
have appeared prior to \cite{DS85} together with their deep connection to integrable systems.
They were identified with the so-called second Poisson structure of the
$N$-th (or generalized) KdV equations,
the KdV equation corresponding to the case of $N=2$.
The first Poisson structure of the KdV equation was found by
\cite{Gar71} and \cite{ZF71}. Later, it was shown by \cite{Mag78},
that there exists not only a Poisson structure but a one-parameter family
of compatible Poisson structures for the KdV equation. In this case
we say that we have a bi-Poisson structure.
The fact that the same equation could be written in two different
Hamiltonian forms has become fundamental in proving integrability of such equations using a recurrence procedure, nowadays called the Lenard-Magri scheme of integrability.

The generalization to the $N$-th KdV hierarchy was suggested by \cite{GD76}
who found a Poisson structure for these
hierarchies, obtained as a Poisson algebra on a suitable space of
scalar differential operators ($N$ being the order of the operators).
A candidate for a second Poisson structure (whence 
a bi-Poisson structure) was conjectured by \cite{Adl79} and
proved by \cite{GD78}. These Poisson structures are
usually known as the \emph{first and second Adler-Gelfand-Dickey (AGD) Poisson structures}
and the second one is a special case of general classical $\mc W$-algebras.

Later on, all the $N$-th KdV hierarchies were ``embedded'' into one big hierarchy
called the KP hierarchy, \cite{Sat81}. A Poisson structure
for this hierarchy (which was a slight modification of the AGD Poisson structure)
was suggested by \cite{Wat83}, and \cite{Dic87} proved that there
exists a second Poisson structure. 
In fact, it was shown by \cite{Rad87} that, for any $N\geq1$,
there is a bi-Poisson structure for the KP hierarchy which induces
the bi-Poisson structure for the $N$-th KdV hierarchy.
We refer to the book of \cite{Dic03} for a detailed exposition of these
topics and a large list of references.

The main goal of the present paper is to give a Poisson vertex algebra
interpretation of the AGD approach to classical $\mc W$-algebras.
This point of view greatly simplifies the theory and also allows to study its matrix generalization.

Let $\mc V$ be a differential algebra with derivation $\partial$.
Following \cite{Adl79},
given a pseudodifferential operator $L\in\mc V((\partial^{-1}))$,
we define the corresponding \emph{Adler map} $A^{(L)}:\,\mc V((\partial^{-1}))\to\mc V((\partial^{-1}))$ 
given by \eqref{adlermap},
and the associated skewadjoint matrix differential operator $H^{(L)}$ with coefficients in $\mc V$,
given by \eqref{20131022:eq1}.
In Theorem \ref{L-c} 
we consider the algebra $\mc V_N^\infty$ of differential polynomials 
on the differential variables $\{u_i\}_{i\geq-N}$ and the pseudodifferential operator
$L=\partial^N+u_{-N}\partial^{N-1}+u_{-N+1}\partial^{N-2}+\dots$,
and we prove that, in this case,
the corresponding 1-parameter family of matrices $H^{(L-c)}$ 
gives a bi-Poisson structure on $\mc V_N^\infty$.
Theorem \ref{L-c-2} is the analogous result
for the algebra $\mc V_N$ of differential polynomials in finitely many variables $\{u_i\}_{i=-N}^{-1}$
and $L=\partial^N+u_{-N}\partial^{N-1}+\dots+u_{-1}$.
Also in this case we get a bi-Poisson structure on $\mc V_N$.
This is the classical $\mc W$-algebra associated to the Lie algebra $\mf{gl}_N$
and its principal nilpotent element.
Furthermore, performing Dirac's reduction by the constraint $u_{-N}=0$,
we obtain bi-Poisson structures on 
$\mc W_N^\infty=\quot{\mc V_N^\infty}{\langle u_{-N}\rangle}$
and on $\mc W_N=\quot{\mc V_N}{\langle u_{-N}\rangle}$
(the latter is the classical $\mc W$-algebra for $\mf{sl}_N$).

In Section \ref{sec:hierarchies}
we associate to each of these bi-Poisson structures
the corresponding integrable hierarchy of bi-Hamiltonian equations,
using the Lenard-Magri scheme:
on $\mc W_N^\infty$ we obtain the KP hierarchy,
while on $\mc W_N$ we obtain the $N$-th KdV hierarchy.

Furthermore, in Section \ref{sub:extra2} we use the Adler map
to define natural PVA homomorphisms
$\mc V_{M+N}^\infty\to\mc V_M^\infty\otimes\mc V_N^\infty$
and $\mc V_{M+N}\to\mc V_M\otimes\mc V_N$,
which are a generalization of the usual ``Miura map'', \cite{Miu68}.
This allows us to give another proof of the Theorem of \cite{KW81}.

A major advantage of our approach is that all the above constructions and proofs
extend in a straightforward way to the case
when the variables $u_i$ are replaced by $m\times m$ matrices $U_i$.
This is done in Section \ref{sec:matrixAGD},
where, as a result, we construct 
the matrix KP and $N$-th KdV bi-Hamiltonian equations.
The main difference with the scalar case is that,
in the matrix case, Dirac reduction (developed by \cite{DSKV13c})
by the constraint $U_{-N}=0$ leads to a non-local bi-Poisson structure.
These non-local Poisson structures have been discovered by \cite{Bil95} and \cite{OS98}.
We then need to use the theory of non-local Poisson structures introduced by \cite{DSK13},
and the machinery of rational matrix pseudodifferential operators developed by \cite{CDSK13b},
to prove integrability of the matrix $N$-th KdV hierarchy.

\section{Preliminaries}\label{sec:back}

\subsection{Some simple facts on formal distribution calculus}\label{sub:formal}
We briefly review here some basic facts on formal distribution calculus which will be used
throughout the paper, cf. \cite{Kac96}.

Given a vector space $\mc A$,
an $\mc A$-valued formal distribution in $z$ is a series of the form
$a(z)=\sum_{n\in\mb Z}a_nz^n$, where $a_n\in\mc A$.
They form a vector space denoted by $\mc A[[z,z^{-1}]]$.
An $\mc A$-valued formal distribution in two variables $z$ and $w$ is a series of the form
$a(z,w)=\sum_{m,n\in\mb Z}a_{mn}z^mw^n$, where $a_{mm}\in\mc A$.
They form a vector space denoted by $\mc A[[z,z^{-1},w,w^{-1}]]$.

The $\delta$-\emph{function} is, by definition, the $\mb F$-valued formal distribution
$$
\delta(z-w)
=\sum_{n\in\mb Z}z^{-n-1}w^n\in\mb F[[z,z^{-1},w,w^{-1}]]\,.
$$
For every $a(z)\in\mc A[[z,z^{-1}]]$ we have
\begin{equation}\label{deltaprop}
a(z)\delta(z-w)
=a(w)\delta(z-w)\,.
\end{equation}
In particular, $\res_za(z)\delta(z-w)=a(w)$,
where $\res_z$ denotes the coefficient of $z^{-1}$.

We denote by $i_z$ the power series expansion for large $|z|$.
For example, 
$$
i_z(z-w)^{-1}=\sum_{k\in\mb Z_+}z^{-k-1}w^k
\,.
$$
Using this notation,
the $\delta$-function can be rewritten as follows
\begin{equation}\label{delta}
\delta(z-w)=i_{z}(z-w)^{-1}-i_{w}(z-w)^{-1}\,.
\end{equation}
For $a(z)=\sum_{n\in\mb Z}a_nz^n\in\mc A[[z,z^{-1}]]$, we 
denote $a(z)_+=\sum_{n\in\mb Z_+}a_nz^n$
and $a(z)_-=\sum_{n<0}a_nz^n$.
It is easy to check that
\begin{equation}
\label{20130927:cor1}
\res_za(z)i_{z}(z-w)^{-1}
=a(w)_+
\,.
\end{equation}

\subsection{The algebra of matrix pseudodifferential operators}\label{sub:pseudo}

By a \emph{differential algebra} we mean a unital commutative associative algebra
$\mc A$ over a field $\mb F$ of characteristic $0$, with a derivation $\partial$.

The algebra $\Mat_{m\times m}\mc A$ of matrices with coefficients in $\mc A$
is a unital associative algebra,
with the obvious action of $\partial$ as a derivation.
We consider the algebra $\Mat_{m\times m}\mc A((\partial^{-1}))$ of $m\times m$ \emph{matrix
pseudodifferential operators} with coefficients in $\mc A$.
Its product is determined by the following formula 
($n\in\mb Z$, $A\in\Mat_{m\times m}\mc A$):
$$
\partial^n\circ A=\sum_{k\in\mb Z_+}\binom{n}{k}A^{(k)}\partial^{n-k}\,.
$$
We say that a non-zero
$A(\partial)=\sum_{n\leq N}A_n\partial^n\in\Mat_{m\times m}\mc A((\partial^{-1}))$
has \emph{order} $\ord(A(\partial))=N$ if $A_N\neq0$.
We denote by
$\Mat_{m\times m}\mc A((\partial^{-1}))_N$
the space of all matrix pseudodifferential operators
of order less than or equal to $N$.
We say that $A(\partial)$ is \emph{monic} if $A_N=\mbb1_m$.

The \emph{residue} of 
$A(\partial)\in\Mat_{m\times m}\mc A((\partial^{-1}))$
is, by definition,
\begin{equation}\label{eq:res}
\res_{\partial} A(\partial)=A_{-1}\ (=\mbox{ coefficient of }\partial^{-1})\,.
\end{equation}
The \emph{adjoint} of $A(\partial)\in\Mat_{m\times m}\mc A((\partial^{-1}))$
is, by definition,
$$
A(\partial)^*=\sum_{n\leq N}(-\partial)^n\circ A_n^t\,,
$$
where $A_n^t$ denotes the transpose matrix of $A_n$.

We have the subalgebras
$\Mat_{m\times m}\mc A[\partial]$ 
of \emph{matrix differential operators}, and
$\Mat_{m\times m}\mc A[[\partial^{-1}]]\partial^{-1}$
of \emph{matrix integral operators}.
Every matrix pseudodifferential operator $A(\partial)$
can be decomposed uniquely as
$A(\partial)=A(\partial)_++A(\partial)_-$,
where $A(\partial)_+\in\Mat_{m\times m}\mc A[\partial]$
and $A(\partial)_-\in\Mat_{m\times m}\mc A[[\partial^{-1}]]\partial^{-1}$.


We let
$\mb F[\partial,\partial^{-1}]\circ\Mat_{m\times m}\mc A\subset\Mat_{m\times m}\mc A((\partial^{-1}))$
be the space of all pseudodifferential operators of the form
$\sum_{n\in\mb Z}\partial^n\circ A_n$,
where all but finitely many elements $A_n\in\Mat_{m\times m}\mc A$ are zero.
We also let $\Mat_{m\times m}\mc A[[\partial,\partial^{-1}]]\supset\Mat_{m\times m}\mc A((\partial^{-1}))$
be the space of all formal series in $\partial,\partial^{-1}$
with coefficients in $\Mat_{m\times m}\mc A$ (it is not an algebra).

Recall the following simple facts about matrix pseudodifferential operators.
\begin{proposition}\label{prop:roots}
Let $A(\partial)$ be a monic matrix pseudodifferential operator of order $N\geq1$.
\begin{enumerate}[(a)]
\item
There exists exactly one monic matrix pseudodifferential operator of order one,
which we denote $A^{\frac{1}{N}}(\partial)$, such that
$(A^{\frac{1}{N}}(\partial))^N=A(\partial)$.
\item
There exists 
exactly one monic matrix pseudodifferential operator of order $-N$, which we denote 
$A^{-1}(\partial)$, such that $A(\partial)\circ A^{-1}(\partial)=A^{-1}(\partial)\circ A(\partial)=\mbb1_m$.
\end{enumerate}
\end{proposition}
The \emph{symbol} of a matrix pseudodifferential operator 
$A(\partial)=\sum_{n\leq N}A_n\partial^n\in\Mat_{m\times m}\mc A((\partial^{-1}))$ 
is the Laurent series $A(z)=\sum_{n\leq N}A_nz^n\in\Mat_{m\times m}\mc A((z^{-1}))$,
where $z$ is an indeterminate commuting with $\mc A$.
This gives us a bijective map $\Mat_{m\times m}\mc A((\partial^{-1}))\longrightarrow\Mat_{m\times m}\mc A((z^{-1}))$
which is not an algebra homomorphism.
For $A(\partial),B(\partial)\in\Mat_{m\times m}\mc A((\partial^{-1}))$, we have
\begin{equation}\label{eq:mult_symbol}
(A\circ B)(z)=A(z+\partial)B(z)\,.
\end{equation}
Here and further, for any $n\in\mb Z$, we expand $(z+\partial)^n$ 
in non-negative powers of $\partial$ (and we let the powers of $\partial$ act to the right,
on the coefficients of $B(z)$),
i.e. by $(z+\partial)^n$ we mean $i_z(z+\partial)^n$.

The following lemma will be used in Section \ref{sec:hierarchies}.
\begin{lemma}\label{lem:residui}
\begin{enumerate}[(a)]
\item
If $a(z),b(z)\in\mc A((z^{-1}))$ are formal Laurent series with coefficients in an algebra $\mc A$,
then
$\res_za(z)b(z-x)=\res_za(z+x)b(z)$ (where we expand $(z\pm x)^n$ in non-negative powers of $x$).
\item
If $A(\partial),B(\partial)\in\mc A((\partial^{-1}))$ are 
pseudodifferential operators over the differential algebra $\mc A$,
then
$\res_zA(z)B^*(-z+\lambda)=\res_zA(z+\lambda+\partial)B(z)$.
\end{enumerate}
\end{lemma}
\begin{proof}
Part (a) follows from the formula of  integration by parts,
$\res_z A(z)\partial_z B(z)=-\res_zB(z)\partial_zA(z)$,
and the Taylor expansion $a(z+x)=e^{x\partial_z}a(z)$.
Part (b) is a special case of (a), when $x=\lambda+\partial$,
acting on the coefficients of $B$.
\end{proof}

\subsection{Poisson vertex algebras}\label{sub:pva}

\begin{definition}\label{def:lambda}
A $\lambda$-\emph{bracket} on a differential algebra $\mc V$ 
is an $\mb F$-linear map
$\{\cdot\,_\lambda\,\cdot\}:\,\mc V\otimes\mc V\to\mc V[\lambda]$
satisfying ($f,g,h\in\mc V$)
\begin{enumerate}[(i)]
\item
\emph{sesquilinearity}:
$\{\partial f_\lambda g\}=-\lambda\{f_\lambda g\}$,
$\{f_\lambda\partial g\}=(\lambda+\partial)\{f_\lambda g\}$,
\item
\emph{Leibniz rules}:
$\{f_\lambda gh\}=\{f_\lambda g\}h+\{f_\lambda h\}g$,
$\{fh_\lambda g\}=\{f_{\lambda+\partial}g\}_{\rightarrow}\!h+\{h_{\lambda+\partial}g\}_{\rightarrow}\!f$.
\end{enumerate}
\end{definition}
Here and further we use the following notation:
if $\{f_\lambda g\}=\sum_{n\in\mb Z_+}\lambda^n c_n$,
then $\{f_{\lambda+\partial}g\}_{\rightarrow}h=\sum_{n\in\mb Z_+}c_n(\lambda+\partial)^nh$.
\begin{definition}\label{def:pva}
A \emph{Poisson vertex algebra} (PVA) is a differential algebra $\mc V$ endowed 
with a $\lambda$-bracket $\{\cdot\,_\lambda\,\cdot\}$
satisfying ($f,g,h\in\mc V$)
\begin{enumerate}[(i)]
\setcounter{enumi}{2}
\item
\emph{skew-symmetry}:
$\{g_\lambda f\}=-\{f_{-\lambda-\partial}g\}$,
\item
\emph{Jacobi identity}:
$\{f_\lambda\{g_\mu h\}\}-\{g_\mu\{f_\lambda h\}\}=\{\{f_\lambda g\}_{\lambda+\mu}h\}$.
\end{enumerate}
\end{definition}
In the skew-symmetry we use the following notation:
if $\{f_\lambda g\}=\sum_{n\in\mb Z_+}\lambda^n c_n$,
then $\{f_{-\lambda-\partial}g\}=\sum_{n\in\mb Z_+}(-\lambda-\partial)^nc_n$.
\begin{definition}\label{def:CFTtype}
Let $\mc V$ be a PVA. A \emph{Virasoro element} $T\in\mc V$ with \emph{central charge} $c\in\mb F$
is such that
$\{T_\lambda T\}=(2\lambda+\partial)T+c\lambda^3$,
a $T$-\emph{eigenvector} $a\in\mc V$ of \emph{conformal weight} $\Delta_a\in\mb F$
is such that
$\{T_\lambda a\}=(\Delta_a\lambda+\partial)a+O(\lambda^2)$,
and a \emph{primary element} $a\in\mc V$ of conformal weight $\Delta_a$
is such that $\{T_\lambda a\}=(\Delta_a\lambda+\partial)a$.
By definition, a \emph{PVA of CFT} (conformal field theory) \emph{type} is generated,
as a differential algebra, by a Virasoro element 
and a finite number of primary elements.
\end{definition}

Let $\mc I\subset\mc V$ be a differential ideal.
We say that $\mc I$ is a \emph{PVA ideal} if $\{\mc I_\lambda\mc V\}\subset\mc I[\lambda]$ 
(by skew-symmetry and the fact that $\mc I$ is a differential ideal we also have
$\{\mc V_\lambda\mc I\}\subset\mc I[\lambda]$).
Given a PVA ideal $\mc I$, we can consider the induced Poisson
vertex algebra structure over the differential algebra $\quot{\mc V}{\mc I}$,
which is called the \emph{quotient PVA}. 

In this paper we consider PVA structures on 
the algebra $R_I=\mb F[u_i^{(n)}\mid i\in I,n\in\mb Z_+]$ 
of differential polynomials in the variables $\{u_i\}_{i\in I}$,
where $I$ is an index set (possibly infinite).
The derivation $\partial$ is defined by $\partial(u_i^{(n)})=u_i^{(n+1)}$, $i\in I,n\in\mb Z_+$.
Note that on $R_I$ we have the following commutation relations:
$\left[\frac{\partial}{\partial u_i^{(n)}},\partial\right]=\frac{\partial}{\partial u_i^{(n-1)}}$,
where the RHS is considered to be zero if $n=0$.

\begin{theorem}[{\cite[Theorem 1.15]{BDSK09}}]\label{master}
Let $\mc V=R_I$
and $H=\big(H_{ij}(\lambda)\big)_{i,j\in I}\in\Mat_{I\times I}\mc V[\lambda]$.
\begin{enumerate}[(a)]
\item 
There is a unique 
$\lambda$-bracket $\{\cdot\,_\lambda\,\cdot\}_H$ on $\mc V$,
such that $\{u_i{}_\lambda u_j\}_H=H_{ji}(\lambda)$ for every $i,j\in I$,
and it is given by the following Master Formula
\begin{equation}\label{masterformula}
\{f_\lambda g\}_H=
\sum_{\substack{i,j\in I\\m,n\in\mb Z_+}}\frac{\partial g}{\partial u_j^{(n)}}(\lambda+\partial)^n
H_{ji}(\lambda+\partial)(-\lambda-\partial)^m\frac{\partial f}{\partial u_i^{(m)}}\,.
\end{equation}
\item 
The $\lambda$-bracket \eqref{masterformula} on $\mc V$
is skew-symmetric if and only if skew-symmetry
holds on generators ($i,j\in I$):
\begin{equation}\label{skewsimgen}
\{u_i{}_\lambda u_j\}_H=-\{u_j{}_{-\lambda-\partial}u_i\}_H\,,
\end{equation}
and this is equivalent to skewadjointness of $H$.
\item 
Assuming that the skew-symmetry condition \eqref{skewsimgen} holds, 
the $\lambda$-bracket \eqref{masterformula} satisfies the Jacobi identity, 
thus making $\mc V$ a PVA, provided that the Jacobi identity holds 
on any triple of generators ($i,j,k\in I$):
\begin{equation}\label{jacobigen}
\{u_i{}_\lambda\{u_j{}_\mu u_k\}_H\}_H
-\{u_j{}_\mu\{u_i{}_\lambda u_k\}_H\}_H
=
\{{\{u_i{}_\lambda u_j\}_H}_{\lambda+\mu}u_k\}_H
\,.
\end{equation}
\end{enumerate}
\end{theorem}
\begin{remark}
Theorem \ref{master} holds in the more general situation when $\mc V$
is any algebra of differential functions in the variables $\{u_i\}_{i\in I}$
(see \cite{BDSK09} for a definition), but in the present paper we will only consider
the case when $\mc V=R_I$.
\end{remark}
\begin{definition}\label{hamop}
A \emph{Poisson structure} on $\mc V$ is a matrix differential operator
$H(\partial)\in\Mat_{\ell\times\ell}\mc V[\partial]$
such that the corresponding $\lambda$-bracket $\{\cdot\,_\lambda\,\cdot\}_H$
defines a PVA structure on $\mc V$.
\end{definition}
\begin{example}\label{gfzN}
On $R_1=\mb F[v,v^{\prime},v^{\prime\prime},\dots]$,
we have the so-called \emph{Gardner-Faddeev-Zakharov} (GFZ) PVA ,
given by
\begin{equation}\label{gfz_bracket}
\{ v_\lambda v \} = \lambda\,,
\end{equation}
and the corresponding Poisson structure is $H(\partial)=\partial$.
In general, for $N\geq1$, 
we consider the algebra of differential polynomials in $N$ variables
$R_N=\mb F[v_i^{(n)}\mid i\in\{1,\dots,N\},n\in\mb Z_+]$
with the generalized GFZ $\lambda$-bracket
%
$$
\{ {v_j}_\lambda {v_i} \} = s_{ij}\lambda
\,\,,\,\,\,\,
i,j\in\{1,\dots,N\}\,,
$$
where $S=(s_{ij})_{i,j=1}^N$ is a symmetric matrix over $\mb F$.
The corresponding Poisson structure is $H(\partial)=S\partial$.
\end{example}
\begin{example}\label{vir-mag}
On $R_1=\mb F[u,u^{\prime},u^{\prime\prime},\dots]$ we have 
the \emph{Virasoro-Magri} PVA, given by ($c\in\mb F$)
\begin{equation}\label{virasoro_bracket}
\{u_\lambda u\}=(\partial+2\lambda)u+c\lambda^3\,,
\end{equation}
with Poisson structure $H(\partial)=u^\prime+2u\partial+c\partial^3$.
\end{example}

\subsection{Hamiltonian equations}\label{sub:ham_structures}

The following proposition is immediate to check:
\begin{proposition}\label{pvahamop}
Let $\mc V$ be a PVA. 
Then we have a well defined Lie algebra
bracket on the quotient space $\quot{\mc V}{\partial\mc V}$:
\begin{equation}\label{lambda=0}
\{\tint f,\tint g\}=\tint \left.\{f_\lambda g\}\right|_{\lambda=0}\,.
\end{equation}
Here and further, $\tint:\mc V\to\quot{\mc V}{\partial\mc V}$ is the canonical quotient map.
Moreover, we have a well defined Lie algebra action of $\quot{\mc V}{\partial\mc V}$ on $\mc V$
by derivations of the commutative associative product on $\mc V$, commuting with $\partial$,
given by 
$$
\{\tint f,g\}=\{f_\lambda g\}|_{\lambda=0}\,.
$$
\end{proposition}
\begin{definition}\label{hamsys}
Let $\mc V$ be a PVA.
The \emph{Hamiltonian equation} 
with \emph{Hamiltonian functional} $\tint h\in\quot{\mc V}{\partial\mc V}$ is
\begin{equation}\label{hameq}
\frac{du}{dt}
=\{\tint h,u\}
\end{equation}
An \emph{integral of motion} 
for the Hamiltonian equation \eqref{hameq} is an element $\tint f\in\quot{\mc V}{\partial\mc V}$
such that $\{\tint h,\tint f\}=0$
(or, equivalently, $\frac{d}{dt}\tint f=0$).
Equation \eqref{hameq} is called \emph{integrable}
if there exists an infinite sequence $\tint h_0=\tint h,\,\tint h_1,\,\tint h_2,\dots$,
of linearly independent integrals of motion in involution: $\{\tint h_m,\tint h_n\}=0$ for all $m,n\in\mb Z_+$.
The corresponding \emph{integrable hierarchy of Hamiltonian equations} is
\begin{equation}\label{eq:hierarchy}
\frac{du}{dt_n}=\{\tint h_n,u\}\,,\,\,\ n\in\mb Z_+\,.
\end{equation}
\end{definition}
In the special case when $\mc V=R_\ell$
and $H\in\Mat_{\ell\times\ell}\mc V[\partial]$ is a Poisson structure,
the Lie bracket \eqref{lambda=0} on $\quot{\mc V}{\partial\mc V}$ takes the usual form 
(see \eqref{masterformula}):
$$
\{\tint f,\tint g\}_H
=\sum_{i,j\in I}\int\frac{\delta g}{\delta u_j}H_{ji}(\partial)\frac{\delta f}{\delta u_i}
\,,
$$
where $\frac{\delta f}{\delta u_i}$ denotes 
the \emph{variational derivative} of $f\in\mc V$ with respect to $u_i$,
\begin{equation}\label{eq:def_varder}
\frac{\delta f}{\delta u_i}=\sum_{n\in\mb Z_+}(-\partial)^n\frac{\partial f}{\partial u_i^{(n)}}\,,
\end{equation}
(it is well defined on $\quot{\mc V}{\partial\mc V}$ since $\frac{\delta}{\delta u_i}\circ\partial=0$),
and the Hamiltonian equation associated to the Hamiltonian functional $\tint h\in\quot{\mc V}{\partial\mc V}$
is, as usual,
$$
\frac{du_i}{dt}
=\sum_{j\in I}H_{ij}(\partial)\frac{\delta h}{\delta u_j},\ i\in I
\,.
$$

\subsection{Bi-PVA and the Lenard-Magri scheme of integrability}
\label{sub:lenard_scheme}
\begin{definition}\label{def:compatible}
Two PVA $\lambda$-brackets $\{\cdot\,_\lambda\,\cdot\}_0$ and $\{\cdot\,_\lambda\,\cdot\}_1$
on a differential algebra $\mc V$ are \emph{compatible}
if any their linear combination (or, equivalently, their sum) is a PVA $\lambda$-bracket.
We say in this case that $\mc V$ is a \emph{bi-PVA}.
\end{definition}
For example, the GFZ $\lambda$-bracket \eqref{gfz_bracket} 
and the Virasoro-Magri $\lambda$-bracket \eqref{virasoro_bracket}
are compatible.

According to the \emph{Lenard-Magri scheme of integrability}, \cite{Mag78},
in order to obtain an integrable hierarchy of Hamiltonian equations,
one needs to find a sequence 
$\{\tint h_n\}_{n\in\mb Z_+}\subset\quot{\mc V}{\partial\mc V}$ 
spanning an infinite dimensional space,
such that
\begin{equation}\label{lenardrecursion}
\{\tint h_n,u\}_1
=\{\tint h_{n+1},u\}_0
\,\,\text{ for }\,\, n\in\mb Z_+,\,u\in\mc V
\,.
\end{equation}
If this is the case, then 
$\{\tint h_m,\tint h_n\}_\delta=0$, for all $m,n\in\mb Z_+$, $\delta=0,1$.
Hence, we get the corresponding integrable hierarchy of Hamiltonian equations \eqref{eq:hierarchy}.
Moreover, if $\{\tint h_0,u\}_0=0$,
and $\{\tint g_n\}_{n\in\mb Z_+}\subset\quot{\mc V}{\partial\mc V}$ 
is another sequence satisfying the Lenard-Magri recursion
$\{\tint g_n,u\}_1=\{\tint g_{n+1},u\}_0$, for all $n\in\mb Z_+$ and $u\in\mc V$,
then the two sequences of integrals of motion are compatible:
$\{\tint h_m,\tint g_n\}_\delta=0$, for all $m,n\in\mb Z_+$, $\delta=0,1$, \cite[Sec.2.1]{BDSK09}.

\subsection{Non-local Poisson vertex algebras}
\label{sub:non_local}

For a vector space $V$, we shall use the following notation:
$$
V_{\lambda,\mu}:=V[[\lambda^{-1},\mu^{-1},(\lambda+\mu)^{-1}]][\lambda,\mu]\,,
$$
namely, the quotient of the $\mb F[\lambda,\mu,\nu]$-module
$V[[\lambda^{-1},\mu^{-1},\nu^{-1}]][\lambda,\mu,\nu]$
by the submodule 
$(\nu-\lambda-\mu)V[[\lambda^{-1},\mu^{-1},\nu^{-1}]][\lambda,\mu,\nu]$.
We have the natural embedding 
$\iota_{\mu,\lambda}:\,V_{\lambda,\mu}\hookrightarrow V((\lambda^{-1}))((\mu^{-1}))$
defined by expanding the negative powers of $\nu=\lambda+\mu$
by geometric series in the domain $|\mu|>|\lambda|$.

Let $\mc V$ be a differential algebra.
A \emph{non-local} $\lambda$-\emph{bracket} on $\mc V$ is an $\mb F$-linear map
$\{\cdot\,_\lambda\,\cdot\}:\,\mc V\otimes \mc V\to \mc V((\lambda^{-1}))$
satisfying the sesquilinearity conditions 
and the Leibniz rules,
as in Definition \ref{def:lambda}.
It is called \emph{skew-symmetric} if the skew-symmetry condition in Definition \ref{def:pva} 
holds as well.
The term $\{g_{-\lambda-\partial}f\}$ in the RHS of the skewsymmetry condition
should be interpreted as follows:
we move $-\lambda-\partial$ to the left and
we expand it in non-negative powers of $\partial$,
acting on the coefficients of the $\lambda$-bracket.
Clearly, from skew-symmetry and the left Leibniz rule, we also have the 
right Leibniz rule, which should be interpreted in a similar way.

The non-local skew-symmetric $\lambda$-bracket $\{\cdot\,_\lambda\,\cdot\}$ 
is called \emph{admissible} if
$$
\{f_\lambda\{g_\mu h\}\}\in\mc V_{\lambda,\mu}
\qquad
\text{ for all } f,g,h\in\mc V\,.
$$
Here we are identifying the space $\mc V_{\lambda,\mu}$
with its image in $\mc V((\lambda^{-1}))((\mu^{-1}))$ via the embedding $\iota_{\mu,\lambda}$.
Note that, from skew-symmetry, 
we also have that 
$\{g_\mu\{f_\lambda h\}\}\in\mc V_{\lambda,\mu}$ 
and $\{\{f_\lambda g\}_{\lambda+\mu} h\}\in\mc V_{\lambda,\mu}$.
Therefore, the Jacobi identity can be understood as an equality in the space $\mc V_{\lambda,\mu}$.
\begin{definition}[\cite{DSK13}]\label{20130513:def}
A \emph{non-local Poisson vertex algebra} (PVA) is a differential algebra $\mc V$
endowed with an admissible non-local skew-symmetric $\lambda$-bracket
$\{\cdot\,_\lambda\,\cdot\}:\,\mc V\otimes \mc V\to \mc V((\lambda^{-1}))$
satisfying the Jacobi identity.
\end{definition}

\subsection{Dirac reduction}
\label{sub:non_local_2}

Let $\mc V$ be a (non-local) Poisson vertex algebra 
with $\lambda$-bracket $\{\cdot\,_{\lambda}\,\cdot\}$.
Let $\theta_1,\dots,\theta_m$ be elements of $\mc V$,
and let $\mc I=\langle\theta_1,\dots,\theta_m\rangle_{\mc V}$
be the differential ideal generated by them.
Consider the matrix pseudodifferential operator
$C(\partial)=(C_{\alpha\beta}(\partial))_{\alpha,\beta=1}^m
\in\Mat_{m\times m}\mc V((\partial^{-1}))$,
whose symbol is
\begin{equation}\label{eq:C}
C_{\alpha\beta}(\lambda)=\{\theta_{\beta}{}_{\lambda}\theta_{\alpha}\}\,.
\end{equation}
By the skew-symmetry condition,
the pseudodifferential operator $C(\partial)$ is skewadjoint.
We shall assume that the matrix pseudodifferential operator $C(\partial)$ is invertible, 
and we denote its inverse by
$C^{-1}(\partial)=\big((C^{-1})_{\alpha\beta}(\partial)\big)_{\alpha,\beta=1}^m$.
\begin{definition}\label{20130514:def}
The \emph{Dirac modification} of the PVA $\lambda$-bracket $\{\cdot\,_{\lambda}\,\cdot\}$,
associated to the elements $\theta_1,\dots,\theta_m$,
is the map
$\{\cdot\,_{\lambda}\,\cdot\}^D:\,\mc V\times\mc V\to\mc V((\lambda^{-1}))$
given by ($a,b\in\mc V$):
\begin{equation}\label{eq:dirac}
\{a_{\lambda}b\}^D
=\{a_{\lambda}b\}
-\sum_{\alpha,\beta=1}^m
\{{\theta_{\beta}}_{\lambda+\partial}b\}_{\to}
(C^{-1})_{\beta\alpha}(\lambda+\partial)
\{a_{\lambda}\theta_{\alpha}\}\,.
\end{equation}
\end{definition}
\begin{theorem}[{\cite[Theorem 2.2]{DSKV13c}}]\phantomsection\label{prop:dirac}
\begin{enumerate}[(a)]
\item
The Dirac modification $\{\cdot\,_\lambda\,\cdot\}^D$
is a PVA $\lambda$-bracket on $\mc V$.
\item
All the elements $\theta_i,\,i=1,\dots,m$, are central 
with respect to the Dirac modified $\lambda$-bracket:
$\{a_\lambda\theta_i\}^D=\{{\theta_i}_\lambda a\}^D=0$
for all $i=1,\dots,m$ and $a\in\mc V$.
\item
The differential ideal $\mc I=\langle\theta_1,\dots,\theta_m\rangle_{\mc V}\subset\mc V$,
generated by $\theta_1,\dots,\theta_m$,
is such that
$\{\mc I\,_\lambda\,\mc V\}^D$,
$\{\mc V\,_\lambda\,\mc I\}^D\,
\subset\mc I((\lambda^{-1}))$.
\end{enumerate}
The quotient space $\quot{\mc V}{\mc I}$ is a (non-local) PVA,
with $\lambda$-bracket induced by $\{\cdot\,_\lambda\,\cdot\}^D$,
which we call the \emph{Dirac reduction} of $\mc V$ 
by the constraints $\theta_1,\dots,\theta_m$.
\end{theorem}
\begin{example}\label{gfz_reduced}
Let $N\geq1$ and let us consider the generalized GFZ PVA $R_N$ defined 
in Example \ref{gfzN} (associated to the symmetric matrix $S=\big(s_{ij}\big)_{i,j=1}^N$).
Let $\theta=v_1+\dots+v_N\in R_N$. We have $\{\theta_\lambda\theta\}=\lambda s$, where
$s=\sum_{h,k=1}^N s_{hk}\in\mb F$.
Provided that $s\neq0$, we can
consider the Dirac reduction of $R_N$ by $\theta$,
and we can identify it with
$R_{N-1}=\mb F[v_i^{(n)}\mid i\in\{1,\dots,N-1\},n\in\mb Z_+]$. 
The corresponding Dirac modified $\lambda$-bracket \eqref{eq:dirac} on $R_{N-1}$ is given by
\begin{equation}\label{formula}
\{{v_i}_\lambda v_j\}^D=\left(s_{ij}-\frac{s_is_j}{s}\right)\lambda\,,
\end{equation}
for all $i,j=1,\ldots,N-1$, where $s_i=\sum_{k=1}^Ns_{ik}\in\mb F$.
\end{example}
\begin{lemma}\label{20131002:lem1}
Let $\mc V_1$ and $\mc V_2$ be PVAs and let
$\varphi:\mc V_1\rightarrow\mc V_2$
be a PVA homomorphism.
Let $\theta_1,\dots\theta_m\in\mc V_1$
be such that the matrix $C(\partial)\in\Mat_{m\times m}\mc V_1((\partial^{-1}))$
defined in \eqref{eq:C} is invertible.
Then, $\varphi$ induces a PVA homomorphism
of the corresponding Dirac reduced PVAs
$$
\varphi:\,\quot{\mc V_1}{\langle\theta_1,\dots,\theta_m\rangle_{\mc V_1}}
\to\quot{\mc V_2}{\langle\varphi(\theta_1),\dots,\varphi(\theta_m)\rangle_{\mc V_2}}
\,.
$$
\end{lemma}
\begin{proof}
Let $D_{\alpha\beta}(\lambda)
=\{\varphi(\theta_\beta)_{\lambda}\varphi(\theta_\alpha)\}_2=\varphi(C_{\alpha\beta}(\lambda))$, 
$\alpha,\beta=1,\dots,m$.
Since $\varphi$ is a differential algebra homomorphism,
the matrix $D(\partial)\in\Mat_{m\times m}\mc V_2((\partial^{-1}))$
is invertible, and $D^{-1}(\partial)=\varphi\big(C^{-1}(\partial)\big)$.
Hence, by \eqref{eq:dirac} we have
$$
\varphi\left(\{a_\lambda b\}_1^D\right)
=\{\varphi(a)_{\lambda}\varphi(b)\}_2^D\,,
$$
for all $a,b\in\mc V_1$, as required.
\end{proof}
In general, if we have two compatible PVA $\lambda$-brackets
$\{\cdot\,_\lambda\,\cdot\}_0$ and $\{\cdot\,_\lambda\,\cdot\}_1$ on $\mc V$
(recall Definition \ref{def:compatible}),
and we take their Dirac reductions by a finite number of constraints
$\theta_1,\dots,\theta_m$,
we do NOT get compatible PVA $\lambda$-brackets 
on $\quot{\mc V}{\mc I}$, where $\mc I=\langle\theta_1,\dots,\theta_m\rangle_{\mc V}$.
However,
in the special case when the constraints $\theta_1,\dots,\theta_m$
are central with respect to the first $\lambda$-bracket $\{\cdot\,_\lambda\,\cdot\}_0$
we have the following result.
\begin{theorem}[{\cite[Theorem 2.3]{DSKV13c}}]\label{20130516:thm1}
Let $\mc V$ be a differential algebra,
endowed with two compatible PVA $\lambda$-brackets 
$\{\cdot\,_\lambda\,\cdot\}_0$, $\{\cdot\,_\lambda\,\cdot\}_1$.
Let $\theta_1,\dots,\theta_m\in\mc V$ be central elements
with respect to the first $\lambda$-bracket:
$\{a_\lambda\theta_i\}_0=0$ for all $i=1,\dots,m$, $a\in\mc V$.
Let $C(\partial)=\big(C_{\alpha,\beta}(\partial)\big)_{\alpha,\beta=1}^m$
be the matrix pseudodifferential operator
given by \eqref{eq:C} for the second $\lambda$-bracket:
$C_{\alpha,\beta}(\lambda)=\{{\theta_\beta}_\lambda{\theta_\alpha}\}_1$.
Suppose that the matrix $C(\partial)$ is invertible,
and consider the Dirac modified PVA $\lambda$-bracket
$\{\cdot\,_\lambda\,\cdot\}_1^D$ given by \eqref{eq:dirac}.
Then,
$\{\cdot\,_\lambda\,\cdot\}_0$
and $\{\cdot\,_\lambda\,\cdot\}_1^D$
are compatible PVA $\lambda$-brackets on $\mc V$.
Moreover, the differential algebra ideal
$\mc I=\langle\theta_1,\dots,\theta_m\rangle_{\mc V}$
is a PVA ideal for both the $\lambda$-brackets 
$\{\cdot\,_\lambda\,\cdot\}_0$
and $\{\cdot\,_\lambda\,\cdot\}_1^D$,
and we have the induced compatible PVA $\lambda$-brackets on $\quot{\mc V}{\mc I}$.
\end{theorem}

\section{Adler type pseudodifferential operators,
classical \texorpdfstring{$\mc W$}{W}-algebras, and the Miura map}\label{sec:AGD}

\subsection{The Adler map for a scalar pseudodifferential operator}\label{subsec:adler}

Let $L$ be a scalar pseudodifferential operator of order $\ord(L)=N\in\mb Z$,
with coefficients in a differential algebra $\mc A$.
The corresponding \emph{Adler map} 
$A^{(L)}:\mc A((\partial^{-1}))\longrightarrow
\mc A((\partial^{-1}))$ is given by (cf. \cite{Adl79})
\begin{equation}\label{adlermap}
A^{(L)}(F)=(LF)_+L-L(F L)_+
=L(FL)_--(LF)_-L
\,,
\end{equation}
for any $F\in\mc A((\partial^{-1}))$.
By the last expression in \eqref{adlermap}, we have that
$A^{(L)}(F)\in\mc A((\partial^{-1}))_{N-1}$ for every $F\in\mc A((\partial^{-1}))$.
Moreover, if $F\in\mc A((\partial^{-1}))_{-N-1}$, then $(FL)_+=(LF)_+=0$,
and therefore $A^{(L)}(F)=0$.
In conclusion, $A^{(L)}$ induces a map
$A^{(L)}:\,\quot{\mc A((\partial^{-1}))}{\mc A((\partial^{-1}))_{-N-1}}\longrightarrow\mc A((\partial^{-1}))_{N-1}$.
Note that $\mc A((\partial^{-1}))=\mb F[\partial,\partial^{-1}]\circ\mc A+\mc A((\partial^{-1}))_{-N-1}$,
and $(\mb F[\partial,\partial^{-1}]\circ\mc A)\cap\mc A((\partial^{-1}))_{-N-1}
=\partial^{-N-1}\mb F[\partial^{-1}]\circ\mc A$.
Hence, we can canonically identify 
$$
\quot{\mc A((\partial^{-1}))}{\mc A((\partial^{-1}))_{-N-1}}\simeq
\quot{\left(\mb F[\partial,\partial^{-1}]\circ\mc A\right)}
{\left(\partial^{-N-1}\mb F[\partial^{-1}]\circ\mc A\right)}
\,,
$$
and we get the induced map
\begin{equation}\label{adlermap2}
A^{(L)}:
\quot{\left(\mb F[\partial,\partial^{-1}]\circ\mc A\right)}
{\left(\partial^{-N-1}\mb F[\partial^{-1}]\circ\mc A\right)}
\longrightarrow\mc A((\partial^{-1}))_{N-1}\,.
\end{equation}

Let $I=\{-N,-N+1,-N+2,\dots\}\subset\mb Z$.
We have the identifications
\begin{equation}\label{id1bis}
\quot{\left(\mb F[\partial,\partial^{-1}]\circ\mc A\right)}{\left(\partial^{-N-1}
\mb F[\partial^{-1}]\circ\mc A\right)}\stackrel{\sim}{\longrightarrow}\mc A^{\oplus I}
\,\,,\,\,\,\,
\sum_{n=-N}^M\partial^n\circ F_n\mapsto (F_n)_{n\in I}\,,
\end{equation}
and
\begin{equation}\label{id2bis}
\mc A((\partial^{-1}))_{N-1}\stackrel{\sim}{\longrightarrow}\mc A^I
\,\,,\,\,\,\,
\sum_{n=-N}^\infty P_n\partial^{-n-1}\mapsto(P_n)_{n\in I}\,.
\end{equation}
Therefore the map $A^{(L)}$ in \eqref{adlermap2}
induces a map
\begin{equation}\label{def:H}
H^{(L)}:\, \mc A^{\oplus I}\longrightarrow \mc A^I\,.
\end{equation}
This map is given by an $I\times I$ matrix differential operator
$H^{(L)}=\big(H^{(L)}_{ij}(\partial)\big)_{i,j\in I}$,
which we compute explicitly
in terms of the generating series of its entries:
\begin{equation}\label{Hseries}
H^{(L)}(\partial)(z,w)=\sum_{i,j\in I}H_{ij}^{(L)}(\partial)z^{-i-1}w^{-j-1}
\,.
\end{equation}
\begin{lemma}\label{hseries}
We have
\begin{equation}\label{h}
H^{(L)}(\partial)(z,w)
=L(w)i_w(w-z-\partial)^{-1}\circ L(z)
-L(z+\partial)i_w(w-z-\partial)^{-1}\circ L^*(-w+\partial)
\,,
\end{equation}
where, as usual, we expand $L(z+\partial)$ and $L^*(-w+\partial)$ in non-negative powers of $\partial$.
\end{lemma}
\begin{proof}
By definition, the matrix element $H^{(L)}_{ij}(\partial)$ is given by
\begin{equation}\label{20131022:eq1}
H^{(L)}_{ij}(\partial)(f)
=\res_{\partial}\left(A^{(L)}\left(\partial^j\circ f\right)\partial^i\right)\,,
\end{equation}
for all $f\in\mc A$ and $i,j\in I$.
Note that, by the above observations on the Adler map $A^{(L)}$,
the RHS of \eqref{20131022:eq1} is zero for $i$ or $j$ less than $-N$.
Therefore, recalling the definition \eqref{delta} of the $\delta$-function,
\begin{equation}\label{20131022:eq2}
\begin{array}{l}
\displaystyle{
H^{(L)}(\partial)(z,w)f
=\sum_{i,j\in\mb Z}\res_{\partial}\left(A^{(L)}
\left(\partial^j\circ f\right)\partial^i\right)z^{-i-1}w^{-j-1}=
} \\
\displaystyle{
=\res_{\partial}\!
\Big(
\big(L(\partial)\delta(w-\partial)\circ f\big)_+
L(\partial)
\delta(z-\partial)
-L(\partial)
\big(\delta(w-\partial)\circ fL(\partial)\big)_+
\delta(z-\partial)
\Big).
}
\end{array}
\end{equation}
By equation \eqref{deltaprop}, we have
\begin{equation}\label{20131022:eq3}
\left(L(\partial)\delta(w-\partial)\circ f\right)_+
=L(w)\delta(w-\partial)_+\circ f=L(w)i_w(w-\partial)^{-1}\circ f\,.
\end{equation}
Similarly,
\begin{equation}\label{20131022:eq4}
\left(\delta(w-\partial)\circ fL(\partial)\right)_+
=\delta(w-\partial)_+\circ\left(L^*(-w+\partial)f\right)
=i_w(w-\partial)^{-1}\circ\left(L^*(-w+\partial)f\right)\,,
\end{equation}
where $(L^*(-w+\partial)f)\in\mc A((w^{-1}))$ is obtained by applying the (non-negative) 
powers of $\partial$ to $f$
(we put parentheses to denote this).
Combining equations \eqref{20131022:eq2}, \eqref{20131022:eq3} and \eqref{20131022:eq4},
we get
\begin{equation}\label{20131108:eq1}
\begin{array}{c}
H^{(L)}(\partial)(z,w)f
=\res_{\partial}\Big(
L(w)i_w(w-\partial)^{-1}\circ fL(\partial)\delta(z-\partial)
\\
-L(\partial)i_w(w-\partial)^{-1}\circ\left(L^*(-w+\partial)f\right)\delta(z-\partial)
\Big)\,.
\end{array}
\end{equation}
By equation \eqref{deltaprop},
inside the residue  we can replace $\partial$, written on the right, by $z$
(and therefore $\partial$, written anywhere, by $z+\partial$, written in the same place).
Hence, equation \eqref{20131108:eq1} gives \eqref{h}.
\end{proof}
In the last term of equation \eqref{h} $L^*(\partial)$ denotes the adjoint of the
pseudodifferential operator $L(\partial)$,
and $L^*(-w+\partial)\in\mc A[\partial]((w^{-1}))$ 
is obtained by replacing $\partial$ by $-w+\partial$
and expanding in non-negative powers of $\partial$.
This is not the same as the adjoint of the differential operator $L(-w+\partial)$.
In fact, we have the identity
\begin{equation}\label{20131022:eq5}
L(z+\partial)^*=L^*(-z+\partial)\,.
\end{equation}

Note that, while $H^{(L)}(\partial)(z,w)$ lies in $\mc A[\partial]((z^{-1},w^{-1}))$,
i.e. it has powers of $z$ and $w$ simultaneously bounded above (by construction),
the two terms in the RHS of \eqref{h} do not:
they lie in $\mc A[\partial]((z^{-1}))((w^{-1}))$
(and the powers of $z$ are not bounded above).

\subsection{Preliminary properties of the Adler map}
\label{sec:2.2}

Let $\mc V$ be a differential algebra,
let $L(\partial)\in\mc V((\partial^{-1}))$ be a pseudodifferential operator of order $N$,
and let $H^{(L)}(\partial)(z,w)$ be as in \eqref{h}.
%
%
\begin{lemma}\phantomsection\label{20130925:lem1}
\begin{enumerate}[(a)]
\item
$H^{(L)}(\partial)(z,w)
=-H^{(L)}(\partial)^*(w,z)$.
\item
The following identity holds:
\begin{equation}
\label{20130925:eq_jacobi}
\begin{array}{l}
\displaystyle{
H^{(L)}(\lambda)(z_2,z_1)
i_{z_2}(z_2-z_3-\mu-\partial)^{-1}
L(z_3)
} \\
\displaystyle{
-H^{(L)}(\lambda)(z_3+\mu+\partial,z_1)
i_{z_2}(z_2-z_3-\mu-\partial)^{-1}
L^*(-z_2+\mu)
} \\
\displaystyle{
-L(z_1)
i_{z_1}(z_1-z_3-\lambda-\mu-\partial)^{-1}
H^{(L)}(\mu)(z_3,z_2)
} \\
\displaystyle{
+L(z_3+\lambda+\mu+\partial)
i_{z_1}(z_1-z_3-\lambda-\mu-\partial)^{-1}
\Big(\Big|_{x=\lambda+\mu+\partial}
\!\!\!
H^{(L)}(\mu)(z_1-x,z_2)
\Big)
} \\
\displaystyle{
=H^{(L)}(\lambda+\mu+\partial)(z_3,z_1)
i_{z_1}(z_1-z_2-\lambda-\partial)^{-1}
L(z_2)
} \\
\displaystyle{
-H^{(L)}(\lambda+\mu+\partial)(z_3,z_2+\lambda+\partial)
i_{z_1}(z_1-z_2-\lambda-\partial)^{-1}L^*(-z_1+\lambda)\,.
}
\end{array}
\end{equation}
\end{enumerate}
In the fourth term of the RHS of equation \eqref{20130925:eq_jacobi}
we used the following notation:
for a Laurent series $P(z)=\sum_{n=-\infty}^Nc_nz^n\in\mc A((z^{-1}))$ and elements $a,b\in\mc A$,
we let
\begin{equation}\label{20131024:eq1}
a\Big(\Big|_{x=\nu+\partial}P(z+x)b\Big)
=
\sum_{n=-\infty}^Nai_z(z+\nu+\partial)^n(c_nb)
\,.
\end{equation}
\end{lemma}
\begin{proof}
Taking the adjoint of \eqref{h}
and using equation \eqref{20131022:eq5}, we get
\begin{equation}\label{h2}
H^{(L)}(\partial)^*(z,w)
=L(z)i_w(w-z+\partial)^{-1}L(w)
-L(w+\partial)i_w(w-z+\partial)^{-1}L^*(-z+\partial)
\,.
\end{equation}
Combining equations \eqref{h} and \eqref{h2}, and using equation \eqref{delta}, we get,
\begin{equation}\label{20131108:eq2}
\begin{array}{l}
H^{(L)}(\partial)(z,w)
+H^{(L)}(\partial)^*(w,z)
\\
=L(w)\delta(w-z-\partial)L(z)
-L(z+\partial)\delta(w-z-\partial)L^*(-w+\partial)
\,.
\end{array}
\end{equation}
By equation \eqref{deltaprop},
the first term in the RHS of \eqref{20131108:eq2} is equal to $L(z+\partial)\delta(w-z-\partial)L(z)$.
Moreover, if $L(z)=\sum_na_nz^n$,
then $L^*(-w+\partial)=\sum_n(w-\partial)^na_n$.
Therefore, by equation \eqref{deltaprop},
 $\delta(w-z-\partial)L^*(-w+\partial)=\delta(w-z-\partial)\sum_nz^na_n=\delta(w-z-\partial)L(z)$.
Hence, 
the second term in the RHS of \eqref{20131108:eq2} is equal to $L(z+\partial)\delta(w-z-\partial)L(z)$
as well, proving part (a).

Using \eqref{h}, we can rewrite each of the six terms of equation \eqref{20130925:eq_jacobi}.
The first term is
\begin{equation}\label{eq1}
\begin{array}{l}
\vphantom{\Big(}
\displaystyle{
H^{(L)}(\lambda)(z_2,z_1)
i_{z_2}(z_2-z_3-\mu-\partial)^{-1}
L(z_3)
} \\
\vphantom{\Big(}
\displaystyle{
= i_{z_1}(z_1-z_2-x)^{-1}
i_{z_2}(z_2-z_3-y)^{-1}
\Big(
L(z_1)
\big(\big|_{x=\lambda+\partial}L(z_2)\big)
\big(\big|_{y=\mu+\partial}L(z_3)\big) 
} \\
\vphantom{\Big(}
\displaystyle{
-L(z_2+x)
\big(\big|_{x=\lambda+\partial}L^*(-z_1+\lambda)\big)
\big(\big|_{y=\mu+\partial}L(z_3)\big)
\Big)
\,,}
\end{array}
\end{equation}
the second term is
\begin{equation}\label{eq2}
\begin{array}{l}
\displaystyle{
H^{(L)}(\lambda)(z_3+\mu+\partial,z_1)
i_{z_2}(z_2-z_3-\mu-\partial)^{-1}
L^*(-z_2+\mu)
} \\
\vphantom{\Big(}
\displaystyle{
=i_{z_1}(z_1-z_3-x-y)^{-1}
i_{z_2}(z_2-z_3-y)^{-1}
\times
} \\
\vphantom{\Big(}
\displaystyle{
\times\Big(
L(z_1)
\big(\big|_{x=\lambda+\partial}L(z_3+y)\big)
\big(\big|_{y=\mu+\partial}L^*(-z_2+\mu)\big)
} \\
\vphantom{\Big(}
\displaystyle{
-L(z_3+x+y)
\big(\big|_{x=\lambda+\partial}L^*(-z_1+\lambda)\big)
\big(\big|_{y=\mu+\partial}L^*(-z_2+\mu)\big)
\Big)
\,,}
\end{array}
\end{equation}
the third term is
\begin{equation}\label{eq3}
\begin{array}{l}
\displaystyle{
L(z_1)
i_{z_1}(z_1-z_3-\lambda-\mu-\partial)^{-1}
H^{(L)}(\mu)(z_3,z_2)
} \\
\vphantom{\Big(}
\displaystyle{
=i_{z_1}(z_1\!-z_3\!-\!x\!-\!y)^{-1}
i_{z_2}(z_2-z_3-y)^{-1}
\Big(
L(z_1)
\big(\big|_{x=\lambda+\partial}L(z_2)\big)
\big(\big|_{y=\mu+\partial}L(z_3)\big)
} \\
\vphantom{\Big(}
\displaystyle{
-L(z_1)
\big(\big|_{x=\lambda+\partial}L(z_3+y)\big)
\big(\big|_{y=\mu+\partial}L^*(-z_2+\mu)\big)
\Big)
\,,}
\end{array}
\end{equation}
the fourth term is
\begin{equation}\label{eq4}
\begin{array}{l}
\displaystyle{
L(z_3+\lambda+\mu+\partial)
i_{z_1}(z_1-z_3-\lambda-\mu-\partial)^{-1}
\Big(\Big|_{x=\lambda+\mu+\partial}
H^{(L)}(\mu)(z_1-x,z_2)
\Big)
} \\
\vphantom{\Big(}
\displaystyle{
=
i_{z_1}(z_1-z_3-x-y)^{-1}
i_{z_2}(z_1-z_2-x)^{-1}
L(z_3+x+y)
} \\
\vphantom{\Big(}
\displaystyle{
\Big(
-
\big(\big|_{x=\lambda+\partial}L(z_2)\big)
\big(\big|_{y=\mu+\partial}i_{z_1}L(z_1-x-y)\big)
} \\
\vphantom{\Big(}
\displaystyle{
+
\big(\big|_{x=\lambda+\partial}L^*(-z_1+\lambda)\big)
\big(\big|_{y=\mu+\partial}L^*(-z_2+\mu)\big)
\Big)
\,,}
\end{array}
\end{equation}
the fifth term is
\begin{equation}\label{eq5}
\begin{array}{l}
\displaystyle{
H^{(L)}(\lambda+\mu+\partial)(z_3,z_1)
i_{z_1}(z_1-z_2-\lambda-\partial)^{-1}
L(z_2)
} \\
\vphantom{\Big(}
\displaystyle{
=
i_{z_1}(z_1-z_3-x-y)^{-1}
i_{z_1}(z_1-z_2-x)^{-1}
\Big(
L(z_1)
\big(\big|_{x=\lambda+\partial}L(z_2)\big)
\big(\big|_{y=\mu+\partial}L(z_3)\big)
} \\
\vphantom{\Big(}
\displaystyle{
-L(z_3+x+y)
\big(\big|_{x=\lambda+\partial}L(z_2)\big)
\big(\big|_{y=\mu+\partial}i_{z_1}L(z_1-x-y)\big)
\Big)
\,,}
\end{array}
\end{equation}
and the last term is
\begin{equation}\label{eq6}
\begin{array}{l}
\displaystyle{
H^{(L)}(\lambda+\mu+\partial)(z_3,z_2+\lambda+\partial)
i_{z_1}(z_1-z_2-\lambda-\partial)^{-1}L^*(-z_1+\lambda)
} \\
\vphantom{\Big(}
\displaystyle{
= i_{z_1}(z_1-z_2-x)^{-1} i_{z_2}(z_2-z_3-y)^{-1}
\Big(
L(z_2+x)
\big(\big|_{x=\lambda+\partial}L^*(-z_1+\lambda)\big)
\times
} \\
\vphantom{\Big(}
\displaystyle{
\times
\big(\big|_{y=\mu+\partial}L(z_3)\big)
-
L(z_3+x+y)
\big(\big|_{x=\lambda+\partial}L^*(-z_1+\lambda)\big)
\big(\big|_{y=\mu+\partial}L^*(-z_2+\mu)\big)
\Big)
\,.}
\end{array}
\end{equation}
Combining the second term in the RHS of \eqref{eq1} 
and the first term in the RHS of \eqref{eq6} we get $0$,
and combining the first term in the RHS of \eqref{eq2} 
and the second term in the RHS of \eqref{eq3} we also get $0$.
Next, combining the first terms in the RHS of \eqref{eq1}, \eqref{eq3} and \eqref{eq5},
we get
$$
\begin{array}{l}
\vphantom{\Big(}
\displaystyle{
\Big(
i_{z_1}(z_1-z_2-x)^{-1}
i_{z_2}(z_2-z_3-y)^{-1}
-
i_{z_1}(z_1-z_3-x-y)^{-1}
i_{z_2}(z_2-z_3-y)^{-1}
} \\
\vphantom{\Big(}
\displaystyle{
-
i_{z_1}(z_1-z_3-x-y)^{-1}
i_{z_1}(z_1-z_2-x)^{-1}
\Big)
L(z_1)
\big(\big|_{x=\lambda+\partial}L(z_2)\big)
\big(\big|_{y=\mu+\partial}L(z_3)\big) 
\,,
}
\end{array}
$$
and this is zero by the obvious identity
\begin{equation}\label{obvious}
a^{-1}b^{-1}-i_a(a+b)^{-1}b^{-1}-i_a(a+b)^{-1}a^{-1}=0\,.
\end{equation}
Next,
combining the second terms in the RHS of \eqref{eq2}, \eqref{eq4} and \eqref{eq6},
we get
$$
\begin{array}{c}
\vphantom{\Big(}
\displaystyle{
\Big(
i_{z_1}(z_1-z_3-x-y)^{-1}
i_{z_2}(z_2-z_3-y)^{-1}
+
i_{z_1}(z_1-z_3-x-y)^{-1}
\times
} \\
\vphantom{\Big(}
\displaystyle{
\times
i_{z_2}(z_1-z_2-x)^{-1}
-
i_{z_1}(z_1-z_2-x)^{-1} 
i_{z_2}(z_2-z_3-y)^{-1}
\Big)
\times
} \\
\vphantom{\Big(}
\displaystyle{
\times
L(z_3+x+y)
\big(\big|_{x=\lambda+\partial}L^*(-z_1+\lambda)\big)
\big(\big|_{y=\mu+\partial}L^*(-z_2+\mu)\big)
\,,}
\end{array}
$$
which, by the identity $i_{z_2}(z_1-z_2-x)^{-1}=i_{z_1}(z_1-z_2-x)^{-1}-\delta(z_1-z_2-x)$
and equation \eqref{obvious},
can be rewritten as follow:
\begin{equation}\label{eq7}
\begin{array}{c}
\vphantom{\Big(}
\displaystyle{
-i_{z_1}(z_1-z_3-x-y)^{-1}
\delta(z_1-z_2-x)
L(z_3+x+y)
\times
} \\
\vphantom{\Big(}
\displaystyle{
\times
\big(\big|_{x=\lambda+\partial}L^*(-z_1+\lambda)\big)
\big(\big|_{y=\mu+\partial}L^*(-z_2+\mu)\big)
\,,}
\end{array}
\end{equation}
Next, combining the first term in the RHS of \eqref{eq4}
and the second term in the RHS of \eqref{eq5},
we get the opposite of \eqref{eq7}, proving the claim.
\end{proof}
\begin{lemma}\label{20131029:lem}
If $L(\partial)\in\mc V[\partial]$,
then
$H^{(L)}(\partial)(z,w)\in\mc V[\partial][z,w]$.
\end{lemma}
\begin{proof}
By equation \eqref{h} $H^{(L)}(\partial)(z,w)$ has no negative powers of $z$,
and by Lemma \ref{20130925:lem1}(a) it has no negative powers of $w$.
\end{proof}

\subsection{Adler type scalar pseudodifferential operators}

\begin{definition}\label{def:adler}
Let $\mc V$ be a differential algebra endowed with a $\lambda$-bracket $\{\cdot\,_\lambda\,\cdot\}$.
We call a pseudodifferential operator $L(\partial)\in\mc V((\partial^{-1}))$
of \emph{Adler type} (for the $\lambda$-bracket $\{\cdot\,_\lambda\,\cdot\}$)
if the following identity holds in $\mc V[\lambda]((z^{-1},w^{-1}))$:
\begin{equation}\label{generating}
\{L(z)_\lambda L(w)\}=H^{(L)}(\lambda)(w,z)
\,.
\end{equation}
\end{definition}
\begin{lemma}\label{20131027:prop1}
Let $\mc V$ be a differential algebra,
let $\{\cdot\,_\lambda\,\cdot\}$ be a $\lambda$-bracket on $\mc V$,
and let $L(\partial)\in\mc V((\partial^{-1}))$
be an Adler type pseudodifferential operator.
Then:
\begin{enumerate}[(a)]
\item The following identity holds in $\mc V[\lambda]((z^{-1},w^{-1}))$:
\begin{equation}\label{skew-symseries}
\{L(z)_{\lambda}L(w)\}=-\{L(w)_{-\lambda-\partial}L(z)\}\,.
\end{equation}
\item The following identity holds in
$\mc V[\lambda,\mu]((z_1^{-1},z_2^{-1},z_3^{-1}))$:
\begin{equation}\label{jacobiseries}
\{L(z_1)_{\lambda}\{L(z_2)_{\mu}L(z_3)\}\}
-\{L(z_2)_{\mu}\{L(z_1)_{\lambda}L(z_3)\}\}
=\{\{L(z_1)_{\lambda}L(z_2)\}_{\lambda+\mu}L(z_3)\}\,.
\end{equation}
\end{enumerate}
\end{lemma}
\begin{proof}
In view of equation \eqref{generating},
part (a) follows from Lemma \ref{20130925:lem1}(a).
Let us prove part (b). Using sesquilinearity, the left and right Leibniz rules and equation \eqref{h}
we can rewrite each term of the equation \eqref{jacobiseries} as follows.
The first one is
%
\begin{align}
&\{L(z_1)_{\lambda}\{L(z_2)_{\mu}L(z_3)\}\}\notag\\
&=H^{(L)}(\lambda)(z_2,z_1)
i_{z_2}(z_2-z_3-\mu-\partial)^{-1}
L(z_3)
\label{eq:a}\\
&+L(z_2)
i_{z_2}(z_2-z_3-\lambda-\mu-\partial)^{-1}
H^{(L)}(\lambda)(z_3,z_1)
\label{eq:b}\\
&-H^{(L)}(\lambda)(z_3+\mu+\partial,z_1)
i_{z_2}(z_2-z_3-\mu-\partial)^{-1}
L^*(-z_2+\mu)
\label{eq:c}\\
&-L(z_3+\lambda+\mu+\partial)
i_{z_2}(z_2-z_3-\lambda-\mu-\partial)^{-1}
\Big(\Big|_{x=\lambda+\mu+\partial}
\!\!\!
H^{(L)}(\lambda)(z_2-x,z_1)
\Big)
\,;\label{eq:d}
\end{align}
the second term is
%
%
\begin{align}
&\{L(z_2)_{\mu}\{L(z_1)_{\lambda}L(z_3)\}_{H}\}_{H}\notag\\
&=H^{(L)}(\mu)(z_1,z_2)
i_{z_1}(z_1-z_3-\lambda-\partial)^{-1}
L(z_3)
\label{eq:a'}\\
&+L(z_1)
i_{z_1}(z_1-z_3-\lambda-\mu-\partial)^{-1}
H^{(L)}(\mu)(z_3,z_2)
\label{eq:b'}\\
&-H^{(L)}(\mu)(z_3+\lambda+\partial,z_2)
i_{z_1}(z_1-z_3-\lambda-\partial)^{-1}
L^*(-z_1+\lambda)
\label{eq:c'}\\
&-L(z_3+\lambda+\mu+\partial)
i_{z_1}(z_1-z_3-\lambda-\mu-\partial)^{-1}
\Big(\Big|_{x=\lambda+\mu+\partial}
\!\!\!
H^{(L)}(\mu)(z_1-x,z_2)
\Big)
\,;
\label{eq:d'}
\end{align}
the third term is
%
%
\begin{align}
&\{\{L(z_1)_{\lambda}L(z_2)\}_{H}{}_{\lambda+\mu}L(z_3)\}_{H}
\notag\\
&=H^{(L)}(\lambda+\mu+\partial)(z_3,z_1)
i_{z_1}(z_1-z_2-\lambda-\partial)^{-1}
L(z_2)
\label{eq:a''}\\
&-H^{(L)}(\lambda+\mu+\partial)(z_3,z_2)
i_{z_1}(z_2-z_1-\mu-\partial)^{-1}
L(z_1)
\label{eq:b''}\\
&-H^{(L)}(\lambda+\mu+\partial)(z_3,z_2+\lambda+\partial)
i_{z_1}(z_1-z_2-\lambda-\partial)^{-1}
L^*(-z_1+\lambda)
\label{eq:c''}\\
&+H^{(L)}(\lambda+\mu+\partial)(z_3,z_1+\mu+\partial)
i_{z_1}(z_2-z_1-\mu-\partial)^{-1}
L^*(-z_2+\mu)
\label{eq:d''}\,.
\end{align}
By Lemma \ref{20130925:lem1}(b) we have the following identity
$$
\eqref{eq:a}+\eqref{eq:c}-\eqref{eq:b'}-\eqref{eq:d'}
=\eqref{eq:a''}+\eqref{eq:c''}\,.
$$
Hence, equation \eqref{jacobiseries} is proved once we show that
\begin{equation}\label{20130925:eq_mid}
\eqref{eq:b}+\eqref{eq:d}-\eqref{eq:a'}-\eqref{eq:c'}
=\eqref{eq:b''}+\eqref{eq:d''}\,.
\end{equation}
Using equation \eqref{20130925:eq_jacobi}
with $\lambda$ and $\mu$ exchanged, and with $z_1$ and $z_2$ exchanged,
we get that equation \eqref{20130925:eq_mid} is equivalent to the following equation
$$
\begin{array}{l}
\vphantom{\Big(}
H^{(L)}(\lambda+\mu+\partial)(z_3,z_2)
\delta(z_2-z_1-\mu-\partial)
L(z_1)
\\
\vphantom{\Big(}
-H^{(L)}(\lambda+\mu+\partial)(z_3,z_1+\mu+\partial)
\delta(z_2-z_1-\mu-\partial)
L^*(-z_2+\mu)=0\,,
\end{array}
$$
which holds by the properties of the $\delta$-function.
\end{proof}
\begin{remark}\label{subalgebra}
Let $\mc V$ be a differential algebra, endowed with a $\lambda$-bracket $\{\cdot\,_\lambda\,\cdot\}$.
Let $L(\partial)\in\mc V((\partial^{-1}))$ be a pseudodifferential operator,
and let $\mc U\subset\mc V$ be the differential subalgebra 
generated by the coefficients of $L(\partial)$.
Clearly, if $L(\partial)$ is of Adler type,
then $\{\cdot\,_\lambda\,\cdot\}$ restricts to a $\lambda$-bracket of $\mc U$.
Lemma \ref{20131027:prop1}, together with Theorem \ref{master},
is saying that $\mc U$, with this $\lambda$-bracket, is a PVA.
\end{remark}

\subsection{The generic pseudodifferential operator of order \texorpdfstring{$N$}{N}
and the corresponding AGD bi-PVA \texorpdfstring{$\mc V_N^\infty$}{V_N}}
\label{sec:2.2b}

Let $N$ be a positive integer. 
Consider the algebra of differential polynomials
$\mc V_N^\infty=\mb F[u_i^{(n)}\mid i\in I,n\in\mb Z_+]$,
where, as before, $I=\{-N,-N+1,-N+2,\dots\}$.
The \emph{generic} pseudodifferential operator on $\mc V_N^\infty$ is, by definition,
\begin{equation}\label{lcap}
L(\partial)
=\partial^N+u_{-N}\partial^{N-1}+u_{-N+1}\partial^{-N-2}+\ldots
=\sum_{n\leq N}u_{-n-1}\partial^n\in\mc V_N^\infty((\partial^{-1})),
\end{equation}
where $u_{-N-1}=1$.
For $c\in\mb F$,
let $H^{(L-c)}(\partial)(z,w)\in\mc V_N^\infty[\partial]((z^{-1},w^{-1}))$
be the corresponding series, as in equation \eqref{h}.

Recall by Theorem \ref{master}(a) that a $\lambda$-bracket on $\mc V_N^\infty$
is uniquely determined by assigning
the $\lambda$-brackets on generators, $\{{u_i}_\lambda{u_j}\}$, for all $i,j\in I$,
or, equivalently, their generating series $\{L(z)_\lambda L(w)\}$.
In particular, there exists a unique $\lambda$-bracket on $\mc V_N^\infty$
such that $L(\partial)-c$ is of Adler type:
\begin{equation}\label{generating2}
\{L(z)_\lambda L(w)\}_c=H^{(L-c)}(\lambda)(w,z)\,.
\end{equation}

It is clear from the definition \eqref{adlermap} of the Adler map $A^{(L)}$
that $A^{(L-c)}$ is linear in $c$,
and therefore 
\begin{equation}\label{HK}
H^{(L-c)}(\partial)=H(\partial)-cK(\partial)\,,
\end{equation} 
where $H=H^{(L)}$ is given by \eqref{Hseries} and \eqref{h}, 
and $K\in\Mat_{I\times I}\mc V_N^\infty[\partial]$.
\begin{theorem}\label{L-c}
The 1-parameter family of $\lambda$-brackets 
$\{\cdot\,_\lambda\,\cdot\}_c=\{\cdot\,_\lambda\,\cdot\}_H-c\{\cdot\,_\lambda\,\cdot\}_K$, $c\in\mb F$,
defines a structure of bi-PVA on $\mc V_N^\infty$.
\end{theorem}
\begin{proof}
The statement is a special case of Remark \ref{subalgebra}.
By Definition \ref{hamop},
we have to prove that the $\lambda$-bracket
on $\mc V$ defined by equation \eqref{generating2}
defines a PVA structure on $\mc V$.
By Theorem \ref{master},
this is the same as proving skew-symmetry \eqref{skewsimgen} and Jacobi identity \eqref{jacobigen}
on generators,
which are the same as equations \eqref{skew-symseries} and \eqref{jacobiseries} respectively.
\end{proof}
\begin{definition}\label{agd}
The \emph{AGD bi-PVA} $\mc V_N^\infty$ is given by the 1-parameter family 
of $\lambda$-brackets $\{\cdot\,_\lambda\,\cdot\}_c$, $c\in\mb F$.
The operators $K$ and $H$ are usually called, respectively, the
\emph{first} and the \emph{second Adler-Gelfand-Dickey Poisson structures}.
\end{definition}
Using equation \eqref{h} we get the following explicit formulas for
$\{L(z)_\lambda L(w)\}_c
=\{L(z)_\lambda L(w)\}_H-c\{L(z)_\lambda L(w)\}_K$:
\begin{equation}\label{eq:H}
\begin{array}{c}
\vphantom{\Big(}
\{L(z)_\lambda L(w)\}_H
=L(z)i_z(z-w-\lambda-\partial)^{-1}L(w)
\\
\vphantom{\Big(}
-L(w+\lambda+\partial)i_z(z-w-\lambda-\partial)^{-1}L^*(-z+\lambda)
\end{array}
\end{equation}
and
\begin{equation}\label{eq:K}
\begin{array}{c}
\vphantom{\Big(}
\{L(z)_\lambda L(w)\}_K
=i_z(z-w-\lambda)^{-1}\left(L(z)-L(w+\lambda)\right)
\\
\vphantom{\Big(}
+i_z(z-w-\lambda-\partial)^{-1}\left(L(w)-L^*(-z+\lambda)\right)\,.
\end{array}
\end{equation}
Expanding equations \eqref{eq:H} and \eqref{eq:K}
in powers of $z$ and $w$, we get
the symbols of the Poisson structures $H$ and $K$ ($i,j\in I$):
\begin{equation}\label{HKij}
\begin{array}{l}
\displaystyle{
H_{ji}(\lambda)
=
\sum_{k,\alpha\in\mb Z_+}\binom{k}{\alpha}u_{i-k-1}
(\lambda+\partial)^\alpha u_{j+k-\alpha}
}\\
\displaystyle{
-\sum_{k,\alpha,\beta\in\mb Z_+}(-1)^\alpha\binom{j}{\alpha}
\binom{i-k-1}{\beta}u_{j+k-\alpha}
(\lambda+\partial)^{\alpha+\beta}u_{i-\beta-k-1}
\,,}
\\
\displaystyle{
K_{ji}(\lambda)
=
\epsilon_{ij}
\sum_{k\in\mb Z_+}\left(
\binom{i}{k}(\lambda+\partial)^k-\binom{j}{k}(-\lambda)^k\right)
u_{i+j-k}
\,,}
\end{array}
\end{equation}
where $\epsilon_{ij}=+1$ if $i,j\in\mb Z_+$, $\epsilon_{ij}=-1$ if $i,j<0$,
and $\epsilon_{ij}=0$ otherwise.
Note that the sums are all finite since $u_k=0$ for $k<-N-1$. 
\begin{remark}\label{gc1}
Let $R_+=\oplus_{i\in\mb Z_+}\mb F[\partial]u_i\subset\mc V^\infty_N$
and $R_-=\big(\oplus_{i\in I_-}\mb F[\partial]u_i\big)\oplus\mb F\subset\mc V^\infty_N$.
By the second equation in \eqref{HKij},
$R_+$ and $R_-$ are commuting Lie conformal subalgebras of $\mc V^\infty_N$
for the $\lambda$-bracket $\{\cdot\,_\lambda\,\cdot\}_K$:
$\{{R_{\pm}}_\lambda R_{\pm}\}_K\subset R_{\pm}[\lambda]$,
and $\{{R_{\pm}}_\lambda R_{\mp}\}_K=0$.
In fact, the Lie conformal algebra $R_+$
is isomorphic to the general Lie conformal algebra $\mf{gc}_1$,
via the map $u_i\mapsto J^i$ (see \cite[Sec.2.10]{Kac96}).
Recall from \cite{DSK02} that we can identify $R_+\simeq\mf{gc}_1\simeq\mb F[\partial,x]$
(the space of polynomials in two commuting variables $\partial$ and $x$)
via the isomorphism $u_i^{(n)}\mapsto\partial^n x^i$.
With this identification, the $\lambda$-bracket on $\mf{gc}_1$ becomes:
$$
\{A(\partial,x)_\lambda B(\partial,x)\}_K
=
A(-\lambda,x+\lambda+\partial)B(\lambda+\partial,x)-B(\lambda+\partial,x-\lambda)A(-\lambda,x)\,.
$$
In the same spirit, we can identify 
$$
R_-\simeq\quot{\big(\mb F[\partial,x^{-1}]\oplus\mb F\big)}{\Span_{\mb F}\{\partial^mx^{-N-1-n}-\delta_{m,0}\delta_{n,0}\}_{m,n\in\mb Z_+}}\,,
$$
via the isomorphism $u_i^{(n)}\mapsto\partial^n x^i$,
for all $i\in I_-$ and $n\in\mb Z_+$.
Under this identification, the $K$-$\lambda$-bracket \eqref{HKij} on $R_-$ becomes:
$$
\{A(\partial,x)_\lambda B(\partial,x)\}_K
=
-i_xA(-\lambda,x+\lambda+\partial)B(\lambda+\partial,x)+i_xB(\lambda+\partial,x-\lambda)A(-\lambda,x)\,.
$$
In the same spirit, we can also rewrite the $H$-$\lambda$-bracket \eqref{HKij}.
We let $R=R_+\oplus R_-=\big(\oplus_{i\in I}\mb F[\partial]u_i\big)\oplus\mb F\subset\mc V^\infty_N$.
Note that $R\subset\mc V^\infty_N$
is not a Lie conformal subalgebra for the $H$-$\lambda$-bracket,
since the expression of the $\lambda$-bracket of two generators is quadratic.
On the other hand, we can represent homogeneous polynomials of degree $2$ 
in the variables $u_i^{(n)}$ with polynomials in $\partial_1,x_1^{\pm1},\partial_2,x_2^{\pm1}$,
via the identification $\partial_1^m\partial_2^nx_1^ix_2^j\mapsto u_i^{(m)}u_j^{(n)}$.
With this notation, the $H$-$\lambda$-bracket \eqref{HKij} becomes
$$
\begin{array}{l}
\vphantom{\Big(}
\displaystyle{
\{A(\partial,x)_\lambda B(\partial,x)\}_K
=
A(-\lambda,x_1)B(\lambda+\partial_1+\partial_2,x_2)
i_{x_1}(x_1-x_2-\partial_2-\lambda)^{-1}
} \\
\vphantom{\Big(}
\displaystyle{
-i_{x_2}B(\lambda+\partial_1+\partial_2,x_2-\partial_1-\lambda)
i_{x_1}A(-\lambda,x_1+\partial_1+\lambda)
i_{x_1}(x_1-x_2+\partial_1+\lambda)^{-1}
\,.}
\end{array}
$$
\end{remark}

\subsection{The generic differential operator of order \texorpdfstring{$N$}{N}
and the corresponding AGD bi-PVA \texorpdfstring{$\mc V_N$}{V_N}}
\label{sub:diff_case}

Let $I_-=\{-N,\dots,-1\}$ ($\subset I$).
Consider the subalgebra of differential polynomials
$\mc V_N=\mb F[u_i^{(n)}\mid i\in I_-,n\in\mb Z_+]$ of $\mc V^\infty_N$,
and the differential operator
\begin{equation}\label{lcap2}
L(\partial)
=\partial^N+u_{-N}\partial^{N-1}+\dots+u_{-2}\partial+u_{-1}
\in\mc V_N[\partial]
\,,
\end{equation}
which we call \emph{generic}.
By Lemma \ref{20131029:lem},
for all $c\in\mb F$ the corresponding operator $H^{(L-c)}(\partial)(z,w)$,
given by equation \eqref{h}, lies in $\mc V[\partial][z,w]$,
i.e. $H^{(L-c)}_{ij}(\partial)=0$ unless $i,j\in I_-$.
Assign the (unique) $\lambda$-bracket on $\mc V_N$
such that $L(\partial)-c$ is of Adler type:
\begin{equation}\label{generating2b}
\{L(z)_\lambda L(w)\}_c=H^{(L-c)}(\lambda)(w,z)\,.
\end{equation}
\begin{theorem}\label{L-c-2}
The 1-parameter family of $\lambda$-brackets 
$\{\cdot\,_\lambda\,\cdot\}_c=\{\cdot\,_\lambda\,\cdot\}_H-c\{\cdot\,_\lambda\,\cdot\}_K$, $c\in\mb F$,
defines a structure of bi-PVA on $\mc V_N$.
\end{theorem}
\begin{proof}
The same as the proof of Theorem \ref{L-c}.
\end{proof}
\begin{definition}\label{agd0}
The \emph{AGD bi-PVA} $\mc V_N$ is given by the 1-parameter family 
of $\lambda$-brackets $\{\cdot\,_\lambda\,\cdot\}_c$, $c\in\mb F$.
\end{definition}
The PVA $\mc V_N$ coincides with the classical $\mc W$-algebra
associated to the Lie algebra $\mf{gl}_N$ and its principal nilpotent element, \cite{DS85}.
The explicit formula for the $\lambda$-brackets of generators corresponding to
$H$ and $K$ are the same as \eqref{eq:H} and \eqref{eq:K}.
The corresponding matrix elements $H_{ji}(\lambda)$ and $K_{ji}(\lambda)$
are the same as in \eqref{HKij},
letting $u_i=0$ for $i\geq0$.
Note that $\mc V_N$ is a PVA subalgebra of $\mc V^\infty_N$ 
for the Poisson structure $K$,
but not for the Poisson structure $H$.
\begin{example}\label{20131003:exa1}
For $N=1$, we have $\mc V_1=\mb F[u^{(n)}\mid n\in\mb Z_+]$, and $L(\partial)=\partial+u\in\mc V[\partial]$.
In this case, 
$\{u_{\lambda}u\}_{H}=-\lambda$ and $K=0$.
\end{example}

\subsection{The PVA homomorphism \texorpdfstring{$\varphi_A$}{phi_A}}\label{sub:extra1}

\begin{proposition}\label{prop:halloween1}
Let $\mc V$ be a PVA, with $\lambda$-bracket $\{\cdot\,_\lambda\,\cdot\}$.
Let 
$$
A(\partial)=\partial^N+\sum_{i\in I}a_i\partial^{-i-1}\in\mc V((\partial^{-1}))
\,\, \,\,
\Big(\text{resp.}\,\,
A(\partial)=\partial^N+\sum_{i\in I_-}a_i\partial^{-i-1}\in\mc V[\partial]
\Big)
$$
be an Adler type pseudodifferential (resp. differential) operator
with respect to the $\lambda$-bracket $\{\cdot\,_\lambda\,\cdot\}$.
Then we have a PVA homomorphism
$$
\begin{array}{l}
\vphantom{\bigg(}
\displaystyle{
\varphi_A:\,\mc V_N^\infty=\mb F[u_i^{(n)}\mid i\in I,n\in\mb Z_+]\to\mc V
} \\
\vphantom{\bigg(}
\displaystyle{
\,\,\,\,
\Big(\text{resp.}\,\,
\varphi_A:\,\mc V_N=\mb F[u_i^{(n)}\mid i\in I_-,n\in\mb Z_+]\to\mc V
\Big)
\,,}
\end{array}
$$
from the AGD PVA $\mc V_N^\infty$ (resp. $\mc V_N$) for $c=0$ to $\mc V$,
given by $\varphi_A(u_i)=a_i$, for all $i$.
\end{proposition}
\begin{proof}
The map $\varphi_A$ is defined, in terms of generating series, by $\varphi_A(L(z))=A(z)$.
Hence, 
we only need to check
that $\varphi_A(\{L(z)_\lambda L(w)\}_H)=\{A(z)_\lambda A(w)\}$,
or, equivalently,
$\varphi_A(H^{(L)}(\lambda)(w,z))=H^{(A)}(\lambda)(w,z)$.
This is clear by the definitions of $H^{(L)}$, $H^{(A)}$ and $\varphi_A$.
\end{proof}

\subsection{Product of Adler type pseudodifferential operators
and the generalized Miura map}\label{sub:extra2}

\begin{proposition}\label{prop:halloween2}
Let $\mc V$ be a PVA
and let 
$A(\partial),B(\partial)\in\mc V((\partial^{-1}))$
be Adler type pseudodifferential operators
such that $\{A(z)_\lambda B(w)\}=0$.
Then $A(\partial)\circ B(\partial)$
is an Adler type pseudodifferential operator.
\end{proposition}
\begin{proof}
By \eqref{eq:mult_symbol} we need to show that
$$
\{A(z+\partial)B(z)_{\lambda}A(w+\partial)B(w)\}
=H^{(A\circ B)}(\lambda)(w,z)\,.
$$
By the sesquilinearity and Leibniz rules, and by the assumption that $A$ and $B$
are of Adler type, and that $\{A(z)_\lambda B(w)\}=0$,
we have
$$
\begin{array}{l}
\displaystyle{
\{A(z+\partial)B(z)_{\lambda}A(w+\partial)B(w)\}
}\\
\vphantom{\Big)}
\displaystyle{
=\{A(z+x)_{\lambda+x} A(w+y)\}
\Big(\Big|_{x=\partial }B(z)\Big)
\Big(\Big|_{y=\partial}B(w)\Big)
}\\
\vphantom{\Big)}
\displaystyle{
+A(w+\lambda+\partial)
\{B(z)_{\lambda+\partial}B(w)\}_\rightarrow
A^*(-z+\lambda)
}\\
\vphantom{\Big)}
\displaystyle{
=\left(A(z+\partial)B(z)\right)
i_z(z-w-\lambda-\partial)^{-1}
A(w+\partial)B(w)
}\\
\vphantom{\Big)}
\displaystyle{
-A(w+\lambda+\partial)B(w+\lambda+\partial)
i_z(z-w-\lambda-\partial)^{-1}
B^*(-z+\lambda+\partial)A^*(-z+\lambda)\,,
}
\end{array}
$$
which is the same as $H^{(A\circ B)}(\lambda)(w,z)$.
\end{proof}
Let $M$ and $N$ be positive integers.
Consider the AGD PVAs (for $c=0$)
$\mc V_M^\infty=\mb F[a_i^{(m)}\mid i\geq-M,m\in\mb Z_+]$
(resp. 
$\mc V_M=\mb F[a_i^{(m)}\mid -M\leq i\leq-1,m\in\mb Z_+]$),
$\mc V_N^\infty=\mb F[b_j^{(n)}\mid j\geq-N,n\in\mb Z_+]$
(resp. 
$\mc V_N=\mb F[b_j^{(n)}\mid -N\leq j\leq-1,n\in\mb Z_+]$),
and $\mc V_{M+N}^\infty=\mb F[u_i^{(n)}\mid i\geq-M-N,n\in\mb Z_+]$
(resp. 
$\mc V_{M+N}=\mb F[u_i^{(n)}\mid -M-N\leq j\leq-1,n\in\mb Z_+]$),
and the corresponding
generic Adler type  pseudodifferential (resp. differential) operators
$$
\begin{array}{l}
\displaystyle{
A(\partial)=\partial^M+\sum_{i=-M}^\infty a_i\partial^{-i-1}
\,\,\,\,
\Big(\text{resp.}\,\,
A(\partial)=\partial^M+\sum_{i=-M}^{-1} a_i\partial^{-i-1}
\Big)\,,
} \\
\displaystyle{
B(\partial)=\partial^N+\sum_{j=-N}^\infty b_j\partial^{-j-1}
\,\,\,\,
\Big(\text{resp.}\,\,
B(\partial)=\partial^N+\sum_{j=-N}^{-1} b_j\partial^{-j-1}
\Big)\,,
} \\
\displaystyle{
L(\partial)=\partial^{M+N}+\sum_{i=-M-N}^\infty u_i\partial^{-i-1}
\,\,\,\,
\Big(\text{resp.}\,\,
L(\partial)=\partial^{M+N}+\sum_{i=-M-N}^{-1} u_i\partial^{-i-1}
\Big)
\,.}
\end{array}
$$
\begin{proposition}\label{prop:halloween3}
We have a PVA structure on the algebra of differential polynomials
$$
\begin{array}{l}
\vphantom{\Big)}
\displaystyle{
\mc V_M^\infty\otimes\mc V_N^\infty=\mb F[a_i^{(n)},b_j^{(n)}\mid i\geq-M, j\geq-N,n\in\mb Z_+]
} \\
\displaystyle{
\,\,\,\,
\Big(\text{resp.}\,\,
\mc V_M\otimes\mc V_N
=\mb F[a_i^{(n)},b_j^{(n)}\mid -M\leq i\leq -1,-N\leq j\leq -1,n\in\mb Z_+]
\Big)
\,,}
\end{array}
$$
given on generators by
\begin{equation}\label{halloween:eq1}
\begin{array}{l}
\vphantom{\Big)}
\displaystyle{
\{A(z)_\lambda A(w)\}=H^{(A)}(\lambda)(w,z)
\,\,,\,\,\,\,
\{B(z)_\lambda B(w)\}=H^{(B)}(\lambda)(w,z)
\,,} \\
\displaystyle{
\{A(z)_\lambda B(w)\}=\{B(z)_\lambda A(w)\}=0
\,,}
\end{array}
\end{equation}
and a PVA homomorphism
$$
\mu_{A,B}:\,
\mc V_{M+N}^\infty\to\mc V_M^\infty\otimes\mc V_N^\infty
\,\,\Big(\text{resp. }
\mu_{A,B}:\,
\mc V_{M+N}\to\mc V_M\otimes\mc V_N
\Big)
\,,
$$
such that $\mu_{A,B}(u_i)$
is the coefficient of $\partial^{-i-1}$ in $A(\partial)\circ B(\partial)$.
In terms of generating series,
$$
\mu_{A,B}(L(z))=A(z+\partial)B(z)\,.
$$
\end{proposition}
\begin{proof}
It is immediate to check that formulas \eqref{halloween:eq1}
define a structure of PVA on $\mc V_M^\infty\otimes\mc V_N^\infty$
(it is the tensor product of the PVAs $\mc V_M^\infty$ and $\mc V_N^\infty$).
By construction, $A(\partial)$ and $B(\partial)$
are of Adler type in this PVA.
Hence, by Proposition \ref{prop:halloween2}
$A(\partial)\circ B(\partial)$ is of Adler type as well.
Therefore, by Proposition \ref{prop:halloween1}
we get the corresponding PVA homomorphism,
which is exactly $\varphi_{A\circ B}=\mu_{A,B}$.
\end{proof}
\begin{definition}\label{gener-miura}
We call $\mu_{A,B}$ the \emph{generalized Miura map} of type $(M,N)$.
The same argument as in the proof of Proposition \ref{prop:halloween3}
can be applied to any number of factors,
so we can talk about the \emph{generalized Miura map}
$\mu_{A_1,\dots,A_s}$ of type $(N_1,\dots,N_s)$.
\end{definition}

\subsection{The classical \texorpdfstring{$\mc W$}{W}-algebra
\texorpdfstring{$\mc W_N$}{W_N}}\label{sub:sln_red}

\begin{lemma}\label{20130926:lem1}
In the AGD bi-PVA $\mc V_N^\infty$ (resp. $\mc V_N$) we have:
\begin{enumerate}[(a)]
\item
$\{u_{-N}{}_\lambda L(w)\}_H=L(w)-L(w+\lambda)$;
\item
$\{L(z)_\lambda u_{-N}\}_H=L^*(-z+\lambda)-L(z)$;
\item
$\{u_{-N}{}_{\lambda}u_{-N}\}_H=-N\lambda$;
\item
$\{u_{-N}{}_\lambda L(w)\}_K=\{L(z){}_\lambda u_{-N}\}_K=0$.
\end{enumerate}
\end{lemma}
\begin{proof}
By equation \eqref{eq:H}, we have
$$
\begin{array}{l}
\vphantom{\Big)}
\{u_{-N}{}_{\lambda}L(w)\}_H
=\res_zL(z)i_{z}(z-w-\lambda-\partial)^{-1}z^{-N}L(w)
\\
\vphantom{\Big)}
-L(w+\lambda+\partial)
\res_z i_z(z-w-\lambda-\partial)^{-1}
L^*(-z+\lambda)z^{-N}
\,.
\end{array}
$$
Note that
$L(z)i_{z}(z-w-\lambda-\partial)^{-1}$ and
$i_z(z-w-\lambda-\partial)^{-1}L^*(-z+\lambda)$
have order $N-1$ in $z$ and their leading
coefficient is $u_{-N-1}=1$.
Hence,
$$
\res_z L(z)i_{z}(z-w-\lambda-\partial)^{-1}z^{-N}
=\res_z i_z(z-w-\lambda-\partial)^{-1}L^*(-z+\lambda)z^{-N}
=1
\,.
$$
This proves part (a). Part (b) follows by the skew-symmetry (Lemma \ref{20131027:prop1}(a)).
Part (c) follows from (a).
Finally, for part (d), we have by \eqref{eq:K}
$$
\begin{array}{l}
\vphantom{\Big(}
\{u_{-N}{}_\lambda L(w)\}_K
=\res_zi_z(z-w-\lambda)^{-1}\left(L(z)-L(w+\lambda)\right)z^{-N}
\\
\vphantom{\Big(}
+\res_zi_z(z-w-\lambda-\partial)^{-1}\left(L(w)-L^*(-z+\lambda)\right)z^{-N}
\\
\vphantom{\Big(}
=\res_zi_z(z-w-\lambda)^{-1}L(z)z^{-N}
-\res_zi_z(z-w-\lambda-\partial)^{-1}L^*(-z+\lambda)z^{-N}
\,,
\end{array}
$$
which is zero by the same argument as before.
\end{proof}
Since, by Lemma \ref{20130926:lem1}(c) and (d), $\{{u_{-N}}_\lambda u_{-N}\}_H$ is not zero
and $u_{-N}$ is central with respect to the Poisson structure $K$,
by Theorems \ref{prop:dirac} and \ref{20130516:thm1} we can perform the Dirac reduction
to get a bi-Poisson structure $(H^D,K)$ on 
$\quot{\mc V_N^\infty}{\langle u_{-N}\rangle}\simeq\mb F[u_i^{(n)}\mid-N\neq i\in I, n\in\mb Z_+]
=:\mc W^\infty_N$
(resp. 
$\quot{\mc V_N}{\langle u_{-N}\rangle}\simeq\mb F[u_i^{(n)}\mid-N\neq i\in I_-, n\in\mb Z_+]=:\mc W_N$).
\begin{proposition}\label{20130925:prop1}
We have a local bi-Poisson structure $(H^D,K)$ on 
$\mc W^\infty_N$ (resp. $\mc W_N$),
where $H^D$ is defined, in terms of the generating series of the $\lambda$-brackets on generators,
by
\begin{equation}\label{eq:H_dirac}
\begin{array}{l}
\vphantom{\Big(}
\{L(z)_{\lambda}L(w)\}_{H^D}
=
L(z)i_z(z-w-\lambda-\partial)^{-1}L(w)
\\
\vphantom{\Big(}
-L(w+\lambda+\partial)i_z(z-w-\lambda-\partial)^{-1}L^*(-z+\lambda)
\\
\vphantom{\Big(}
-\frac1N\left(L(w+\lambda+\partial)-L(w)\right)
(\lambda+\partial)^{-1}
\left(L^*(-z+\lambda)-L(z)\right)\,,
\end{array}
\end{equation}
and $K$ is given by \eqref{eq:K}.
\end{proposition}
\begin{proof}
Formula \eqref{eq:H_dirac} follows from \eqref{eq:dirac}  using
Lemma \ref{20130926:lem1}.
Since $L^*(-z+\lambda)=\big(\big|_{x=\lambda+\partial}L(z-x)\big)$,
the term $(\lambda+\partial)^{-1}(L^*(-z+\lambda)-L(z))$ does not contain
negative powers of $\lambda+\partial$, thus proving that the $\lambda$-bracket 
given in \eqref{eq:H_dirac} is local.
\end{proof}
Expanding the equation \eqref{eq:H_dirac} in powers of $z$ and $w$ the matrix elements of
$H^D$ are (cf. equation \eqref{HKij}):
\begin{equation}\label{HKijD}
H^D_{ji}(\lambda)
=
H_{ji}(\lambda)
-\frac1N\sum_{\alpha,\beta\geq1}(-1)^{\alpha}
\binom{j}{\alpha}\binom{i}{\beta}u_{j-\alpha}
(\lambda+\partial)^{\alpha+\beta-1}u_{i-\beta}
\,,
\end{equation}
and $K_{ji}(\lambda)$ is the same as in equation \eqref{HKij}.
Letting $T=u_{-N+1}$, we have
\begin{equation}\label{virfield}
\{T{}_{\lambda}T\}_{H^D}
= H^D_{-N+1,-N+1}(\lambda)
= (2\lambda+\partial)T
+\frac{N^3-N}{12}\lambda^3\,,
\end{equation}
namely, $T$ is a Virasoro
element with central charge $\frac{N^3-N}{12}$ (cf. Definition \ref{def:CFTtype}).
Furthermore, we have
\begin{equation}\label{eq:eigenfields_series}
\{T{}_\lambda u_j\}_{H^D}
=\big(\partial+(N+j+1)\lambda\big)u_j+O(\lambda^2)\,,
\end{equation}
namely, $u_j$ is a $T$-eigenvector of conformal weight $N+j+1$, for every $j\neq-N$.
\begin{definition}\label{wn}
The \emph{classical} $\mc W$-\emph{algebra} $\mc W_N$ is given by the 
bi-Poisson structure $(H^D,K)$.
\end{definition}
\begin{remark}\label{rem:wn}
The classical $\mc W$-algebras $\mc W_N$ 
can be obtained by performing Drinfeld-Sokolov Hamiltonian reduction
for the Lie algebra $\mf{sl}_N$ and its principal nilpotent element, \cite{DS85}.
\end{remark}
\begin{example}\label{w2}
For $N=2$ we have $\mc W_2=\mb F[u^{(n)}\mid n\in\mb Z_+]$,
where $u=u_{-1}$,
and the PVA structure $\{u\,_\lambda\,u\}_c=H^D(\lambda)-cK(\lambda)$, $c\in\mb F$,
given by \eqref{eq:H_dirac} and \eqref{eq:K}, becomes, in this case,
\begin{equation}\label{eq:w2}
\{u{}_{\lambda}u\}_c
=(2\lambda+\partial)u+\frac{1}{2}\lambda^3-2c\lambda\,.
\end{equation}
This is the Virasoro-Magri PVA, which is the classical $\mc W$-algebra 
for the Lie algebra $\mf{sl}_2$ and its principal nilpotent element.
\end{example}
\begin{example}\label{w3}
For $N=3$ we have $\mc W_3=\mb F[u^{(n)},v^{(n)}\mid n\in\mb Z_+]$,
where $u=u_{-2}$ and $v=u_{-1}$,
and the PVA structure is
\begin{equation}\label{eq:w3}
\begin{array}{l}
\vphantom{\Big(}
\displaystyle{
\{u_\lambda u\}_c
=(2\lambda+\partial)u+2\lambda^3 
\,,}\\
\vphantom{\Big(}
\displaystyle{
\{u_\lambda v\}_c
=(3\lambda+\partial)v+u\lambda^2+\lambda^4-3c\lambda 
\,}\\
\vphantom{\Big(}
\displaystyle{
\{v_\lambda v\}_c
=(2\lambda+\partial)\left(\partial v-\frac12\partial^2u-\frac13u^2\right)
-\frac16(2\lambda+\partial)^3u-\frac23\lambda^5 
\,.}
\end{array}
\end{equation}
This is known as the Zamolodchikov PVA, \cite{Zam85},
which is the classical $\mc W$-algebra
for $\mf{sl}_3$ and its principal nilpotent element.
Note that $T=u$ is a Virasoro element with central charge $2$,
and it is not hard to check that $w_3=v-\frac12\partial u$ is a primary element of conformal weight $3$.
In particular, $\mc W_3$ is a PVA of CFT type (cf. Definition \ref{def:CFTtype}).
\end{example}
For arbitrary $N\geq1$, we let $T=u_{-N+1}$, and we have
$$
\begin{array}{l}
\displaystyle{
\{T_{\lambda}T\}_{H^D}
=(2\lambda+\partial)T
+\frac{N^3-N}{12}
\lambda^3
\,,} \\
\displaystyle{
\{T_{\lambda} u_k\}_{H^D}
=((N+1+k)\lambda+\partial) u_k+O(\lambda^2)
\,,\,\,
k=-N+2,\dots,-1
\,.}
\end{array}
$$
Hence, 
$T$ is  a Virasoro element, and $u_k$, $k=-N+2,\dots,-1$, are
$T$-eigenvectors.
It was proved by \cite{BFOFW90} and \cite{DFIZ91} that, in fact,
$\mc W_N$ is a PVA of CFT type.

\subsection{The Kupershmidt-Wilson theorem and the Miura map}
\label{sec:kw}

Consider the generalized GFZ PVA
$R_N=\mb F[v_i^{(n)}\mid i=1,\dots,N,n\in\mb Z_+]$
from Example \ref{gfzN}, 
with $\lambda$-bracket $\{{v_i}_\lambda{v_j}\}=-\delta_{ij}\lambda$
(we are taking $S=-\mbb1_N$),
and its Dirac reduction $\quot{R_N}{\langle v_1+\dots+v_N\rangle}$ 
as in Example \ref{gfz_reduced}.
In this case, the Dirac reduced $\lambda$-bracket among generators \eqref{formula}
reads ($i,j=1,\dots,N-1$):
\begin{equation}\label{formula2}
\{{v_i}_\lambda v_j\}^D=\big(\frac1N-\delta_{ij}\big)\lambda\,.
\end{equation}
Recall from Sections \ref{sub:diff_case} and \ref{sub:sln_red}
the definitions of the AGD PVA $\mc V_N$,
and of the classical $\mc W$-algebra $\mc W_N$, respectively.
In this section we want to give another proof of the following theorem due to \cite{KW81},
that we restate according to our formalism.
\begin{theorem}[Kupershmidt-Wilson Theorem]\label{kw_thm}
\begin{enumerate}[(a)]
\item
We have an injective PVA homomorphism
$\mu:\,\mc V_N\hookrightarrow R_N$, 
given by
\begin{align}\label{miuraop}
\mu(L(\partial))=(\partial+v_N)(\partial+v_{N-1})\cdots(\partial+v_1)
\in R_N[\partial]\,,
\end{align}
where $L(\partial)\in\mc V_N[\partial]$ 
is as in \eqref{lcap2}.
\item
The map $\mu$ in part (a)
induces an injective PVA homomorphism
$$
\mu:\,\mc W_N\hookrightarrow\quot{R_N}{\langle v_1+\dots+v_N \rangle}\,.
$$
\end{enumerate}
\end{theorem}
The map $\mu$ is called the Miura map (\cite{Miu68}).
It allows us to express each differential variable $u_i$
as a differential polynomial in $R_N$.
The original proof was given by \cite{KW81}. 
A shorter proof can be found in the work of \cite{Dic82}.
\begin{proof}
Recall from Example \ref{20131003:exa1} 
that $\mc V_1\simeq R_1=\mb F[v^{(n)}\mid n\in\mb Z_+]$,
with $\{v_\lambda v\}=-\lambda$.
Hence, $R_N=\mc V_1\otimes\dots\otimes\mc V_1$ ($N$ times).
By Proposition \ref{prop:halloween3} and Definition \ref{gener-miura}
we then have the corresponding generalized Miura map of type $(1,1,\dots,1)$,
defined by \eqref{miuraop}.
For the injectiveness of this map (and the induced map on the Dirac reductions) 
we refer to \cite{KW81}.
Part (b) follows immediately from Lemma \ref{20131002:lem1}.
\end{proof}
\begin{remark}\label{20131104:rem}
The aim of the Kupershmidt-Wilson Theorem was to prove that the matrix
differential operator $H^{(L)}(\partial)$, attached to a generic 
differential operator $L$, is a Poisson structure.
Indeed, it was well known that the operator $-\mbb1_N\partial$
is a Poisson structure (cf. Example \ref{gfzN}) on $R_N$.
The Kupershmidt-Wilson Theorem shows that
$(\mc V_N,H^{(L)})\subset(R_N,-\partial\mbb1_N)$ is a PVA subalgebra.
In particular, $H^{(L)}$ must be a Poisson structure.
However, one cannot apply the same argument in the case of a generic pseudodifferential operator,
since one does not have a factorization analogue to \eqref{miuraop}.
\end{remark}
\begin{example}\label{miura2}
Recall from Example \ref{w2} that
$\mc W_2=\mb F[u^{(n)}\mid n\in\mb Z_+]$,
with $\lambda$-bracket \eqref{eq:w2} (with $c=0$).
By Theorem \ref{kw_thm}(b) we have a PVA inclusion 
$\mc W_2\hookrightarrow\quot{R_2}{\langle v_1+v_2\rangle_{R_2}}
\simeq\mb F[v^{(n)}\mid n\in\mb Z_+]$
(where $v=v_1$),
with $\lambda$-bracket $\{v_{\lambda}v\}=-\frac{1}{2}\lambda$.
The Miura map is given by $u=v^\prime-v^2$, \cite{Miu68}.
\end{example}

\section{Gelfand-Dickey integrable hierarchies}\label{sec:hierarchies}
In this section we want to show how to apply the Lenard-Magri scheme
of integrability (see Section \ref{sub:lenard_scheme})
in order to obtain integrable hierarchies for the
the bi-PVAs we constructed in Section \ref{sec:AGD}.

\subsection{Integrable hierarchies for the AGD bi-PVAs $\mc V_N^\infty$ and $\mc V_N$}
\label{sub:hierarchies1}

Recall from Section \ref{sec:2.2b} the definition of the AGD bi-PVA 
$\mc V_N^\infty=\mb F[u_i^{(n)}\mid i\in I,n\in\mb Z_+]$
(as before, $I=\{-N,-N+1,-N+2,\dots\}$),
associated to the generic pseudodifferential operator $L(\partial)$
as in \eqref{lcap},
and recall from Section \ref{sub:diff_case} the AGD bi-PVA 
$\mc V_N=\mb F[u_i^{(n)}\mid i\in I_-,n\in\mb Z_+]$
(where $I_-=\{-N,-N+1,\dots,-1\}$),
associated to the generic differential operator $L(\partial)$
as in \eqref{lcap2}.
Unless otherwise specified, 
throughout this section we let $\mc V=\mc V_N^\infty$ or $\mc V_N$,
with its bi-Poisson structure $(H,K)$,
and we let $L(\partial)$ as in \eqref{lcap} or \eqref{lcap2}.
Let, for $k\geq1$,
\begin{equation}\label{hk}
h_k=\frac Nk\res_z L^{\frac kN}(z)\in\mc V\,,
\end{equation}
where $L^\frac1N(\partial)\in\mc V((\partial^{-1}))$
is uniquely defined by Proposition \ref{prop:roots}.
\begin{theorem}\label{prop:lenard_works}
We have an integrable hierarchy of bi-Hamiltonian equations in $\mc V$:
$$
\frac{du}{dt_k}=\{\tint h_k,u\}_H=\{\tint h_{k+N},u\}_K
\,\,,\,\,\,\,
k\geq1
\,.
$$
\end{theorem}
The remainder of this section will be the proof of Theorem \ref{prop:lenard_works}.

\begin{lemma}\label{lem:15032013}
Let $\mc V$ be an arbitrary differential algebra endowed with a $\lambda$-bracket
$\{\cdot\,_{\lambda}\,\cdot\}$.
Let $L(\partial)\in\mc V((\partial^{-1}))$ be a monic pseudodifferential operator
of order $N>0$.
Then, for all $k\geq1$, the following identity holds in $\mc V((w^{-1}))$:
\begin{equation}\label{eq:lenard1}
\res_z
\left.\{L^{\frac kN}(z)_{\lambda}L(w)\}\right|_{\lambda=0}
=\frac kN
\res_z\{L(z+x)_x L(w)\} \big(\big|_{x=\partial}L^{\frac kN-1}(z)\big)\,.
\end{equation}
\end{lemma}
\begin{proof}
Since, by \eqref{eq:mult_symbol},
$L^{\frac kN}(z)=L^{\frac 1N}(z+\partial)L^{\frac 1N}(z+\partial)\dots L^{\frac 1N}(z)$ 
($k$ times), we have, 
by sesquilinearity and the right Leibniz rule,
\begin{equation}\label{eq:15032013_1}
\{L^{\frac kN}(z)_{\lambda}L(w)\}
=\sum_{l=1}^k\{L^{\frac1N}(z+x)_{\lambda+x+y}L(w)\}
\big(\big|_{x=\partial}L^{\frac{k-l}{N}}(z)\big)
\big(\big|_{y=\partial}(L^*)^{\frac{l-1}{N}}(-z+\lambda)\big)
\,.
\end{equation}
Taking the residue of both sides of equation \eqref{eq:15032013_1} and
using Lemma \ref{lem:residui}(b), we get
$$
\begin{array}{l}
\displaystyle{
\res_z\{L^{\frac kN}(z)_{\lambda}L(w)\}
}
\\
\displaystyle{
=\res_z\sum_{l=1}^k
\{L^{\frac1N}(z+\lambda+x+y)_{\lambda+x+y}L(w)\}
\big(\big|_{x=\partial}L^{\frac{k-l}{N}}(z+\lambda+y)\big)
\big(\big|_{y=\partial}L^{\frac{l-1}{N}}(z)\big),
}
\end{array}
$$
and setting $\lambda=0$ we get
\begin{equation}\label{eq:15032013_1e}
\res_z\{L^{\frac kN}(z)_{\lambda}L(w)\}\Big|_{\lambda=0}
=
k\res_z\{L^{\frac1N}(z+x)_{x}L(w)\}
\big(\big|_{x=\partial}L^{\frac{k-1}{N}}(z))\,.
\end{equation}
On the other hand, letting $k=N$ in \eqref{eq:15032013_1}, we have
\begin{equation}\label{eq:15032013_1b}
\{L(z)_{\lambda}L(w)\}
=\sum_{l=1}^N\{L^{\frac1N}(z+x)_{\lambda+x+y}L(w)\}
\big(\big|_{x=\partial}L^{\frac{N-l}{N}}(z)\big)
\big(\big|_{y=\partial}(L^*)^{\frac{l-1}{N}}(-z+\lambda)\big)
\,.
\end{equation}
If we replace, in equation \eqref{eq:15032013_1b},
$z$ by $z+\partial$ and $\lambda$ by $\lambda+\partial$
acting on $L^{\frac kN-1}(z)$,
we get
\begin{equation}\label{eq:15032013_1c2}
\begin{array}{l}
\displaystyle{
\{L(z+x)_{\lambda+x}L(w)\}\big(\big|_{x=\partial}L^{\frac kN-1}(z)\big)
} \\
\displaystyle{
=\sum_{l=1}^N\{L^{\frac1N}(z+x)_{\lambda+x+y}L(w)\}
\big(\big|_{x=\partial}L^{\frac{k-l}{N}}(z)\big)
\big(\big|_{y=\partial}(L^*)^{\frac{l-1}{N}}(-z+\lambda)\big)
\,.
}
\end{array}
\end{equation}
Taking residues of both sides of equation \eqref{eq:15032013_1c2}
and using Lemma \ref{lem:residui}(b), we get
\begin{equation}\label{eq:15032013_1c}
\begin{array}{l}
\displaystyle{
\res_z \{L(z+x)_{\lambda+x}L(w)\}\big(\big|_{x=\partial}L^{\frac kN-1}(z)\big)
} \\
\displaystyle{
=\res_z\sum_{l=1}^N\{L^{\frac1N}(z+\lambda+x+y)_{\lambda+x+y}L(w)\}
\big(\big|_{x=\partial}L^{\frac{k-l}{N}}(z+\lambda+y)\big)
\big(\big|_{y=\partial}L^{\frac{l-1}{N}}(z)\big)
.
}
\end{array}
\end{equation}
In the second equality we used Lemma \ref{lem:residui}(b).
Setting $\lambda=0$ in both sides of equation \eqref{eq:15032013_1c},
we get
\begin{equation}\label{eq:15032013_1d}
\res_z\{L(z+x)_{x}L(w)\}\big(\big|_{x=\partial}L^{\frac kN-1}(z)\big)
=N\res_z\{L^{\frac1N}(z+x)_{x}L(w)\}
\big(\big|_{x=\partial}L^{\frac{k-1}{N}}(z)\big)
\,.
\end{equation}
Equation \eqref{eq:lenard1} follows from equations \eqref{eq:15032013_1e} and \eqref{eq:15032013_1d}.
\end{proof}
\begin{lemma}\label{prop:var_der}
For every $i\in I$ (resp. $I_-$) and $k\geq1$, we have in $\mc V=\mc V_N^\infty$ (resp. $\mc V_N$),
\begin{equation}\label{eq:varder}
\frac{\delta h_k}{\delta u_i}
=\res_z(z+\partial)^{-i-1}L^{\frac kN-1}(z)\,.
\end{equation}
\end{lemma}
\begin{proof}
Let $\{\cdot\,_\lambda\,\cdot\}$ be any $\lambda$-bracket on $\mc V$,
and let $A_{ij}(\partial)=\{{u_j}_\lambda u_i\}$, $i,j\in I$ or $I_-$,
be the associated matrix differential operator.
Taking the coefficient of $w^{-j-1}$ in both sides of equation \eqref{eq:lenard1} we have,
by the definition \eqref{hk} of $h_k$,
\begin{equation}\label{20131107:eq1}
\begin{array}{l}
\vphantom{\Big(}
\displaystyle{
\{{h_k}_{\lambda}u_j\}\big|_{\lambda=0}
=
\res_z\{L(z+x)_x u_j\} \big(\big|_{x=\partial}L^{\frac kN-1}(z)\big)
} \\
\vphantom{\Big(}
\displaystyle{
=
\sum_{i}\res_z\{{u_i}_x u_j\} (z+x)^{-i-1} \big(\big|_{x=\partial}L^{\frac kN-1}(z)\big)
} \\
\vphantom{\Big(}
\displaystyle{
=
\sum_{i}A_{ji}(\partial)
\res_z(z+\partial)^{-i-1} L^{\frac kN-1}(z)
\,.}
\end{array}
\end{equation}
On the other hand, by the Master Formula \eqref{masterformula} 
and the definition \eqref{eq:def_varder} of the variational derivative, we have
\begin{equation}\label{20131107:eq2}
\{{h_k}_{\lambda}u_j\}\big|_{\lambda=0}
=
\sum_{i}A_{ji}(\partial)\frac{\delta h_k}{\delta u_i}
\,.
\end{equation}
Equation \eqref{eq:varder} follows from \eqref{20131107:eq1} and \eqref{20131107:eq2},
since the matrix differential operator $A(\partial)$ is arbitrary.
\end{proof}
\begin{lemma}\label{20131107:lem}
The local functionals $\tint h_k\in\quot{\mc V}{\partial\mc V},\, k\in\mb Z_+\backslash N\mb Z_+$, 
are all linearly independent.
\end{lemma}
\begin{proof}
By definition, $\mc V$ is a polynomial algebra in the infinitely many variables $u_i^{(n)}$,
for $i\in I$ or $I_-$ and $n\in\mb Z_+$.
Let $\mc V\twoheadrightarrow\mb F[u_{-N}]$, $f\mapsto\bar f$,
be the evaluation homomorphism at $u_i^{(n)}=0$,
where $(i,n)\neq(-N,0)$.
We have
$\overline{L(z)}=z^N+u_{-N}z^{N-1}$, and
$$
\overline{L^{\frac kN-1}(z)}=\sum_{h\in\mb Z_+}\binom{\frac kN-1}{h}(u_{-N})^h z^{k-N-h}
\,.
$$
Hence, by equation \eqref{eq:varder}, we get
$$
\overline{\frac{\delta h_k}{\delta u_{-N}}}
=
\binom{\frac kN-1}{k}(u_{-N})^{k}
\,,
$$
for all $k\geq1$. The claim follows.
\end{proof}
\begin{remark}
In fact, for $\mc V=\mc V_N^\infty$, 
it is not difficult to prove, using equation \eqref{eq:varder},
that $\frac{\delta h_k}{\delta u_i}=0$ for all $i\geq k-N+1$,
and $\frac{\delta h_k}{\delta u_{k-N}}=1$.
Therefore the elements $\frac{\delta h_k}{\delta u}\in{\mc V_N^\infty}^{\oplus I},\, k\geq1$,
and so the elements $\tint h_k\in\quot{\mc V_N^\infty}{\partial\mc V_N^\infty}$, are all linearly independent.
On the other hand, for $\mc V=\mc V_N$,
$L(\partial)$ is a differential operator,
and so $L^{\frac kN}(z)$ has non negative powers of $z$ for all $k\in N\mb Z_+$.
Hence, $\tint h_k=0$ for all $k\in N\mb Z_+$.
\end{remark}
\begin{lemma}\label{cor:17032013}
For the AGD bi-PVA $\mc V=\mc V_N^\infty$ or $\mc V_N$
with bi-Poisson structure $(H,K)$, we have
\begin{enumerate}[(a)]
\item
$\left.\{h_k{}_\lambda L(w)\}_H\right|_{\lambda=0}
=L^{\frac kN}(w+\partial)_+L(w)-L(w+\partial)L^{\frac kN}(w)_+$;
\item
$\left.\{h_k{}_\lambda L(w)\}_K\right|_{\lambda=0}
=L^{\frac{k}{N}-1}(w+\partial)_+L(w)-L(w+\partial)L^{\frac{k}{N}-1}(w)_+$.
\end{enumerate}
\end{lemma}
\begin{proof}
By Lemma \ref{lem:15032013} and equation \eqref{eq:H} we have
\begin{equation}\label{20131106:eq1}
\begin{array}{l}
\displaystyle{
\left.\{h_k{}_\lambda L(w)\}_H\right|_{\lambda=0}
=\res_z
L^{\frac kN}(z)i_z(z-w-\partial)^{-1}
L(w)
}
\\
\displaystyle{
-L(w+\partial)
\res_z L^{\frac{k}{N}-1}(z)i_z(z-w-\partial)^{-1}L^*(-z)
\,.}
\end{array}
\end{equation}
By equation \eqref{20130927:cor1}, we have
\begin{equation}\label{20131106:eq2}
\res_z
L^{\frac kN}(z)i_z(z-w-\partial)^{-1}
=L^{\frac kN}(w+\partial)_+
\,,
\end{equation}
while, by Lemma \ref{lem:residui}(b) 
and equation \eqref{20130927:cor1},
we have
\begin{equation}\label{20131106:eq3}
\begin{array}{l}
\vphantom{\Big(}
\displaystyle{
\res_z L^{\frac{k}{N}-1}(z)i_z(z-w-\partial)^{-1}L^*(-z)
=
\res_z L^{\frac{k}{N}-1}(z+\partial)i_z(z-w)^{-1}L(z)
} \\
\vphantom{\Big(}
\displaystyle{
=
\res_z L^{\frac{k}{N}}(z)i_z(z-w)^{-1}
=
L^{\frac{k}{N}}(w)_+
\,.}
\end{array}
\end{equation}
Combining equations \eqref{20131106:eq1}, \eqref{20131106:eq2} and \eqref{20131106:eq3},
we get part (a).
Similarly, for part (b), 
we use Lemma \ref{lem:15032013} and equation \eqref{eq:K} to get
\begin{equation}\label{20131106:eq4}
\begin{array}{l}
\vphantom{\Big(}
\displaystyle{
\left.\{h_k{}_\lambda L(w)\}_K\right|_{\lambda=0}
=
\res_zi_z(z-w)^{-1}
\big(L(z+\partial)-L(w+\partial)\big)L^{\frac kN-1}(z)
} \\
\vphantom{\Big(}
\displaystyle{
+
\res_z
L^{\frac kN-1}(z)
i_z(z-w-\partial)^{-1}
\big(L(w)-L^*(-z)\big)
\,.}
\end{array}
\end{equation}
By equation \eqref{20130927:cor1} and Lemma \ref{lem:residui}(b) we have
\begin{equation}\label{20131106:eq5}
\begin{array}{l}
\displaystyle{
\res_zi_z(z-w)^{-1}
L(z+\partial)L^{\frac kN-1}(z)
=
L^{\frac kN}(w)_+
} \\
\displaystyle{
=
\res_z
L^{\frac kN-1}(z)
i_z(z-w-\partial)^{-1}
L^*(-z)
\,.}
\end{array}
\end{equation}
Moreover, by equation \eqref{20130927:cor1} we also have
\begin{equation}\label{20131106:eq6}
\res_zi_z(z-w)^{-1}
L(w+\partial)
L^{\frac kN-1}(z)
=
L(w+\partial)L^{\frac kN-1}(w)_+\,,
\end{equation}
and
\begin{equation}\label{20131106:eq7}
\res_z
L^{\frac kN-1}(z)
i_z(z-w-\partial)^{-1}L(w)
=
L^{\frac kN-1}(w+\partial)_+
L(w)
\,.
\end{equation}
Combining equations \eqref{20131106:eq4}, 
\eqref{20131106:eq5}, \eqref{20131106:eq6}, and \eqref{20131106:eq7}, 
we get the claim.
\end{proof}
\begin{lemma}\label{lem:lenard_works}
For every $k\geq1$, 
we have the Lenard-Magri recursion
\begin{equation}\label{leneq}
\{h_k{}_\lambda u)\}_H\big|_{\lambda=0}
=\left.\{h_{k+N}{}_\lambda u\}_K\right|_{\lambda=0}
\,,\,\,
\text{ for all } u\in\mc V
\,.
\end{equation}
\end{lemma}
\begin{proof}
By Lemma \ref{cor:17032013},
the recursion \eqref{leneq} holds for $u=L(w)$,
the generating series of the generators of $\mc V$.
Hence, \eqref{leneq} holds for all $u\in\mc V$  by the Leibniz rule.
\end{proof}
\begin{lemma}\label{lem:lenard_start}
For every $\varepsilon\in\{1,\dots,N\}$, 
we have 
\begin{equation}\label{leneq2}
\{{h_{\varepsilon}}_\lambda u\}_K\big|_{\lambda=0}=0
\,,\,\,
\text{ for all } u\in\mc V
\,.
\end{equation}
\end{lemma}
\begin{proof}
For $1\leq\varepsilon<N$, we have $L^{\frac{\varepsilon}{N}-1}(w)_+=0$,
and therefore equation \eqref{leneq2} holds by Lemma Lemma \ref{cor:17032013}(b).
Moreover, 
$\{h_N{}_\lambda L(w)\}_K\big|_{\lambda=0}
=L(w)-L(w+\partial)\cdot1=0$.
\end{proof}
\begin{proof}[Proof of Theorem \ref{prop:lenard_works}]
According to the Lenard-Magri scheme of integrability (see Section \ref{sub:lenard_scheme}),
by Lemmas \ref{lem:lenard_works} and \ref{lem:lenard_start}
we have that $\tint h_k,\,k\geq1$, are integrals of motion in involution:
$\{\tint h_m,\tint h_n\}_{H,K}=0$ for all $m,n\geq1$.
By Lemma \ref{20131107:lem} they span an infinite dimensional space, as required.
\end{proof}
\begin{remark}\label{20132507:rem1}
It follows from Lemma \ref{cor:17032013} and equation \eqref{eq:mult_symbol}
that the Hamiltonian equation corresponding to the Hamiltonian functional
$\tint h_k$, $k\geq1$, can be written as (in terms of generating series)
\begin{equation}\label{laxpair}
\frac{dL(w)}{dt_k}=[(L^{\frac kN})_+,L](w)\,,
\end{equation}
where on the RHS we have to take the symbol of the usual commutator
of pseudodifferential operators.
This equation is the symbol of the usual \emph{Lax pair} representation
of the AGD hierarchies of Hamiltonian equations.
\end{remark}

\subsection{Integrable hierarchies for the \texorpdfstring{$\mc W$}{W}-algebra 
\texorpdfstring{$\mc W^\infty_N$}{W_N^infty} and \texorpdfstring{$\mc W_N$}{W_N}.}
\label{sub:hierarchies2}

As in the previous section, let $\mc V=\mc V_N^\infty$ or $\mc V_N$.
Let also $(H,K)$ be the AGD bi-Poisson structure on $\mc V$,
and let $H^D$ be the Dirac modification of $H$ by the constraint $\theta=u_{-N}$.
The corresponding $\lambda$-bracket is given, in terms of generating series,
by equation \eqref{eq:H_dirac}.
Recall by Theorem \ref{20130516:thm1} that $(H^D,K)$ is also a bi-Poisson
structure on $\mc V$.
\begin{lemma}\label{prop:17032013}
For any $k\geq1$, we have, in $\mc V$,
$$
\{h_k{}_\lambda L(w)\}_{H}\big|_{\lambda=0}
=\{h_k{}_\lambda L(w)\}_{H^D}\big|_{\lambda=0}\,.
$$
\end{lemma}
\begin{proof}
By Lemma \ref{lem:15032013} and equation \eqref{eq:H_dirac} we have
$$
\begin{array}{l}
\vphantom{\Big(}
\displaystyle{
\{h_k{}_\lambda L(w)\}_{H^D}\big|_{\lambda=0}
-\{h_k{}_\lambda L(w)\}_{H}\big|_{\lambda=0}
} \\
\vphantom{\Big(}
\displaystyle{
=
-\frac1N\left(L(w+\partial)-L(w)\right)\partial^{-1}
\res_z\left(L^*(-z)-L(z+\partial)\right)L^{\frac kN-1}(z)
\,.}
\end{array}
$$
This is zero since, by Lemma \ref{lem:residui}(b), we have
$$
\res_zL^*(-z)L^{\frac kN-1}(z)
=
\res_zL^{\frac kN}(z)
=
\res_zL(z+\partial)L^{\frac kN-1}(z)
\,.
$$
\end{proof}
Recall from Section \ref{sub:sln_red} the definition of the classical $\mc W$-algebras
$\mc W^\infty_N=\mb F[u_i^{(n)}i\in I\setminus\{-N\},n\in\mb Z_+]=\quot{\mc V^\infty_N}{\langle u_{-N}\rangle}$ 
and $\mc W_N=\mb F[u_i^{(n)}\mid i\in I_-\setminus\{-N\},n\in\mb Z_+]=\quot{\mc V_N}{\langle u_{-N}\rangle}$,
obtained from the AGD bi-PVAs $\mc V_N^\infty$ and $\mc V_N$ respectively, 
via Dirac reduction.
We shall denote $\mc W=\mc W_N^\infty$ or $\mc W_N$, with its bi-Poisson structure $(H^D,K)$.

With an abuse of notation, we denote $h_k\in\mc W$, for $k\geq1$,
the image of \eqref{hk} in the quotient space $\mc W=\quot{\mc V}{\langle u_{-N}\rangle}$.
By Lemmas \ref{lem:lenard_works}, \ref{lem:lenard_start}, and \ref{prop:17032013},
we have the Lenard-Magri recursions ($u\in\mc W$):
$$
\begin{array}{l}
\vphantom{\Big(}
\displaystyle{
\{{h_k}_\lambda u\}_K\big|_{\lambda=0}=0
\,\,\text{ for all } k=1,\dots,N
\,,} \\
\vphantom{\Big(}
\displaystyle{
\{h_k{}_\lambda u\}_{H^D}\big|_{\lambda=0}
=\{h_{k+N}{}_\lambda u\}_K\big|_{\lambda=0}
\,\,\text{ for all } k\geq1
\,.}
\end{array}
$$
Furthermore, with the same argument as in the proof of Lemma \ref{20131107:lem},
we get
$$
\begin{array}{l}
\displaystyle{
\overline{\frac{\delta h_k}{\delta u_{-N+1}}}
=
\binom{\frac kN-1}{\frac{k-1}2}(u_{-N+1})^{\frac{k-1}2}\,\text{ if } k \text{ is odd, and } 0 \text{ otherwise }
\,,} \\
\displaystyle{
\overline{\frac{\delta h_k}{\delta u_{-N+2}}}
=
\binom{\frac kN-1}{\frac{k}2-1}(u_{-N+1})^{\frac{k}2-1}\,\text{ if } k
\text{ is even, and } 0 \text{ otherwise }
\,,} 
\end{array}
$$
where this time $f\mapsto\bar f$
denotes the evaluation map $\mc W\twoheadrightarrow\mb F[u_{-N+1}]$
at $u_i^{(n)}=0$ for $(i,n)\neq(-N+1,0)$.
It follows, in particular, that the local functionals $\tint h_k\in\quot{\mc W}{\partial\mc W}$,
for $k\in\mb Z_+\backslash N\mb Z_+$
are linearly independent.
In conclusion, according to the Lenard-Magri scheme of integrability, we get the following
\begin{theorem}\label{prop:lenard_works_W}
We have an integrable hierarchy of Hamiltonian equations in $\mc W$:
$$
\frac{du}{dt_k}=\{\tint h_k,u\}_{H^D}=\{\tint h_{k+N},u\}_K
\,\,,\,\,\,\,
k\geq1
\,.
$$
\end{theorem}

\begin{example}[$\mc W^\infty_1$: the KP hierarchy]\label{exa:kp}
On $\mc W^\infty_1=\mb F[u_i^{(n)}\mid i,n\in\mb Z_+]$,
we have
$L(\partial)=\partial+\sum_{i\in\mb Z_+}u_i\partial^{-i-1}$.
It is not difficult to compute
the first few integrals of motion $\tint h_k$, $k\geq1$, directly from the definition \eqref{hk}:
$$
\tint h_1=\tint u_0
\,,\,\,
\tint h_2=\tint u_1
\,,\,\,
\tint h_3=\tint u_2+u_0^2
\,,\,\,
\tint h_4=\tint u_3+3u_0u_1
\,,\dots
$$
To find the corresponding bi-Hamiltonian equations,
we use Lemma \ref{cor:17032013}.
We have
$L(w)_+=w$,
$L^2(w)_+=w^2+2u_0$,
$L^3(w)_+=w^3+3u_0w+3(u_1+u_0')$.
Hence, 
\begin{equation}\label{KP}
\begin{array}{l}
\displaystyle{
\frac{dL(w)}{dt_1}=\partial L(w)
\,\,,\,\,\,\,
\frac{dL(w)}{dt_2}
=\partial^2L(w)+2w\partial L(w)+2(L(w)-L(w+\partial))u_0
\,,} \\
\displaystyle{
\frac{dL(w)}{dt_3}=
\partial^3L(w)+3w\partial^2L(w)+3w^2\partial L(w)+3u_0\partial L(w)
\,,} \\
\displaystyle{
\,\,\,\,\,\,\,\,\,\,\,\,\,\,\,\,\,\,\,\,\,\,\,\,\,\,\,\,\,\,\,\,\,\,\,\,\,\,\,\,\,\,\,\,\,
+3(L(w)-L(w+\partial))((w+\partial)u_0+u_1)
\,\dots}
\end{array}
\end{equation}
\begin{remark}\label{rem:kp}
Consider the first two equations in the second system of the hierarchy \eqref{KP},
and the first equation in the third system of \eqref{KP}.
After eliminating the variables $u_1$ and $u_2$ 
and relabeling $t_1=y$, $t_2=t$ and $u=2u_0$, we get
\begin{equation}\label{eq:kp}
3u_{yy}=(4u_t-u'''-6uu')'\,,
\end{equation}
which is known as the Kadomtsev-Petviashvili (KP) equation.
\end{remark}
\end{example}

\begin{remark}
In fact, we have infinitely many
bi-Poisson structures for the KP equation,
corresponding to the biPVA's $\mc W_N^\infty=\mb F[u_i^{(n)}\mid i\geq-N+1,n\in\mb Z_+]$, 
for $N\geq1$, \cite{Rad87}.
An explicit differential algebra isomorphism
$\varphi_N:\,\mc W_1^\infty\to\mc W_N^\infty$
is defined by the equation
$$
\varphi_N(L(z))
=
L_N^{\frac1N}(z)
\,,
$$
where $L(\partial)$ is as in Example \ref{exa:kp},
and $L_N(\partial)=\partial^N+\sum_{i\geq-N+1}u_i\partial^{-i-1}$.
This is not a PVA isomorohism (since $L_N^{\frac1N}(\partial)$ is not of Adler type).
On the other hand, one can check that
the integrable hierarchy in $\mc W^{\infty}_N$, given by Theorem \ref{prop:lenard_works_W}, 
is the same for every choice of the positive integer $N$.
Namely, we have, for every $N\geq1$,
$$
\varphi_{N}\bigg(
\frac{d L^N(w)}{d t_{k}}
\bigg)
=\frac{d L_N(w)}{d t_{N,k}}
=[(L_N^{\frac kN})_+,L_N](w)
\,.
$$
\end{remark}

\begin{example}[$\mc W_2$: the KdV hierarchy]
Recall from Example \ref{w2} that $\mc W_2=\mb F[u^{(n)}\mid n\in\mb Z_+]$,
with the bi-PVA structure as in \eqref{eq:w2}.
The first few fractional powers of 
$L(\partial)=\partial^2+u\in\mc W_2[\partial]$ are
$$
\begin{array}{l}
\displaystyle{
L^{\frac12}(\partial)
=\partial+\frac12u\partial^{-1}
-\frac14u'\partial^{-2}+\frac18(u''-u^2)\partial^{-3}
-\frac{1}{16}(u'''-6uu')\partial^{-4}+\dots
\,,} \\
\displaystyle{
L^{\frac32}(\partial)
=\partial^3+\frac32u\partial+\frac34u'
+\frac18(3u^2-u'')\partial^{-1}+\dots
\,,}
\end{array}
$$
from which we get 
$L^{\frac12}(w)_+=w$,
$L^{\frac32}(w)_+=w^3+\frac{3}{4}(2w+\partial)u$,
and 
$\tint h_1=\tint u$, $\tint h_3=\tint\frac14u^2$.
By Lemma \ref{cor:17032013} we get
the corresponding Hamiltonian equations \eqref{eq:hierarchy}:
$\frac{du}{dt_1}=u'$ and the Korteweg-de Vries equation
$$
\frac{du}{dt_3}=\frac14(u'''+6uu')\,.
$$
\end{example}

\begin{example}[$\mc W_3$: the Boussinesq hierarchy]
Recall from Example \ref{w3} that $\mc W_3=\mb F[u^{(n)},v^{(n)}\mid n\in\mb Z_+]$,
with the bi-PVA structure as in \eqref{eq:w3}.
The first few fractional powers of 
$L(\partial)=\partial^3+u\partial+v\in\mc W_3[\partial]$ are
$$
\begin{array}{l}
\displaystyle{
L^{\frac13}(\partial)
=\partial+\frac13u\partial^{-1}
-\frac13(u'-v)\partial^{-2}
+\frac19(2u''-3v'-u^2)\partial^{-3}+\dots
\,,} \\
\displaystyle{
L^{\frac23}(\partial)
=\partial^2+\frac23u\partial+\frac13(2v-u')\partial^{-1}
+\dots\,.}
\end{array}
$$
Hence, 
$L^{\frac13}(w)_+=w$,
$L^{\frac23}(w)_+=w^2+\frac{2}{3}u$,
and $\tint h_1=\tint u$, $\tint h_2=\tint v$.
The corresponding Hamiltonian equations are
$\frac{du}{dt_1}=u'$, $\frac{dv}{dt_1}=v'$, and
$$
\frac{du}{dt_2}=-u''+2v'
\,\,,\,\,\,\,
\frac{dv}{dt_2}=v''-\frac23u'''-\frac23 uu'
\,.
$$
After eliminating $v$ from this system, we get the Boussinesq equation:
$$
u_{tt}=-\frac13\big(u^{(4)}-4(uu')'\big)\,.
$$
\end{example}

\section{Generalization to the matrix case}\label{sec:matrixAGD}

\subsection{Adler type matrix pseudodifferential operators}\label{subsec:adler_matrix}

Let $\mc V$ be a differential algebra,
and let $\mc M=\Mat_{m\times m}\mc V$.
Let $L=\left(L_{ab}(\partial)\right)_{a,b=1}^m\in\mc M((\partial^{-1}))$ 
be a matrix pseudodifferential operator of order $\ord(L)=N\in\mb Z$.
As in the scalar case,
we can define the corresponding Adler map 
$A^{(L)}:\,\mc M((\partial^{-1}))\to\mc M((\partial^{-1}))$ given by \eqref{adlermap},
and with identifications analogous to \eqref{id1bis} and \eqref{id2bis},
we get the corresponding map
$H^{(L)}:\,\mc M^{\oplus I}\to \mc M^I$ (cf. \eqref{def:H}),
where, as before, $I=\{-N,-N+1,\dots\}$.
This map is represented by a tensor
$H^{(L)}=\big(H^{(L)}_{ij;abcd}(\partial)\big)_{i,j\in I;\,a,b,c,d\in\{1,\dots,m\}}$,
where $H^{(L)}_{ij;abcd}(\partial)\in\mc V[\partial]$.
As in Lemma \ref{hseries},
we can write an explicit formula for $H^{(L)}$,
in terms of the generating series
$H^{(L)}_{abcd}(\partial)(z,w)=\sum_{i,j\in I}H_{ij;abcd}^{(L)}(\partial)z^{-i-1}w^{-j-1}$.
We have (cf. equation \eqref{h}):
\begin{equation}\label{h_mat}
\begin{array}{c}
\displaystyle{
H_{abcd}^{(L)}(\partial)(z,w)
=L_{ad}(w)i_w(w-z-\partial)^{-1}\circ L_{cb}(z)
}\\
\displaystyle{
-L_{ad}(z+\partial)i_w(w-z-\partial)^{-1}\circ L_{cb}^*(-w+\partial)
\,.
}
\end{array}
\end{equation}

Let $\mc V$ be a differential algebra endowed with a $\lambda$-bracket $\{\cdot\,_\lambda\,\cdot\}$.
As in the scalar case, 
we say that a matrix pseudodifferential operator $L(\partial)\in\mc M((\partial^{-1}))$
is of \emph{Adler type} (for the $\lambda$-bracket $\{\cdot\,_\lambda\,\cdot\}$)
if the following identity holds in $\mc V[\lambda]((z^{-1},w^{-1}))$:
\begin{equation}\label{generating_mat}
\{L_{ab}(z)_\lambda L_{cd}(w)\}=H_{cdab}^{(L)}(\lambda)(w,z)\,,
\end{equation}
for all $a,b,c,d=1,\dots,m$.
The analogue of Lemma \ref{20130925:lem1}
still holds in the matrix case.
As a consequence, we get (cf. Lemma \ref{20131027:prop1}):
\begin{lemma}\label{20131027:prop1_mat}
Let $\mc V$ be a differential algebra,
let $\{\cdot\,_\lambda\,\cdot\}$ be a $\lambda$-bracket on $\mc V$,
and let $L(\partial)\in\mc M((\partial^{-1}))$ be an Adler type matrix pseudodifferential operator.
Then:
\begin{enumerate}[(a)]
\item The following identity holds in $\mc V[\lambda]((z^{-1},w^{-1}))$:
\begin{equation}\label{skew-symseries_mat}
\{L_{ab}(z)_{\lambda}L_{cd}(w)\}=-\{L_{cd}(w)_{-\lambda-\partial}L_{ab}(z)\}\,,
\end{equation}
for all $a,b,c,d=1,\dots,m$.
\item The following identity holds in
$\mc V[\lambda,\mu]((z_1^{-1},z_2^{-1},z_3^{-1}))$:
\begin{equation}\label{jacobiseries_mat}
\begin{array}{c}
\displaystyle{
\{L_{ab}(z_1)_{\lambda}\{L_{cd}(z_2)_{\mu}L_{ef}(z_3)\}\}
-\{L_{cd}(z_2)_{\mu}\{L_{ab}(z_1)_{\lambda}L_{ef}(z_3)\}\}
}\\
\displaystyle{
=\{\{L_{ab}(z_1)_{\lambda}L_{cd}(z_2)\}_{\lambda+\mu}L_{ef}(z_3)\}\,.
}
\end{array}
\end{equation}
for all $a,b,c,d,e,f=1,\dots,m$.
\end{enumerate}
\end{lemma}

\subsection{The generic matrix pseudodifferential operator of order \texorpdfstring{$N$}{N}
and the corresponding AGD bi-PVA}
\label{sec:2.2b_mat}

For $N,m\geq1$,
let $\mc V_{N,m}^{\infty}=\mb F[u_{i,ab}^{(n)}\mid i\in I,a,b=1,\ldots,m,\,n\in\mb Z_+]$,
and $\mc M_{N,m}^{\infty}=\Mat_{m\times m}\mc V_{N,m}^{\infty}$,
and let $\mc V_{N,m}=\mb F[u_{i,ab}^{(n)}\mid i\in I_-,a,b=1,\ldots,m,\,n\in\mb Z_+]$,
and $\mc M_{N,m}=\Mat_{m\times m}\mc V_{N,m}$.
The \emph{generic} matrix pseudodifferential operator on $\mc V_{N,m}^{\infty}$ 
(resp. $\mc V_{N,m}$) is
\begin{equation}\label{lcap_mat}
\begin{array}{l}
\vphantom{\Big(}
\displaystyle{
L(\partial)
=\partial^N\mbb1_m+U_{-N}\partial^{N-1}+U_{-N+1}\partial^{-N-2}+\ldots
\in\mc M_{N,m}^{\infty}((\partial^{-1}))
\,, } \\
\vphantom{\Big(}
\displaystyle{
\Big(\,\text{ resp. } 
L(\partial)
=\partial^N\mbb1_m+U_{-N}\partial^{N-1}+\ldots+U_{-1}
\in\mc M_{N,m}[\partial]
\Big)
\,.}
\end{array}
\end{equation}
where $U_{i}=(u_{i,ab})_{a,b=1}^m\in\mc M_{N,m}^{\infty}$ for all $i\in I$
(resp. $U_{i}\in\mc M_{N,m}$ for all $i\in I_-$).
There is a unique $\lambda$-bracket $\{\cdot\,_\lambda\,\cdot\}_c$ 
on $\mc V_{N,m}^{\infty}$ (resp. $\mc V_{N,m}$)
such that $L(\partial)-c\mbb1_m$ is of Adler type:
\begin{equation}\label{generating2_mat}
\{L_{ab}(z)_\lambda L_{cd}(w)\}_c=H_{cdab}^{(L-c\mbb1_m)}(\lambda)(w,z)\,,
\end{equation}
for all $a,b,c,d=1,\dots,m$.
Letting, as in Section \ref{sec:2.2b},
$H^{(L-c\mbb1_m)}(\partial)=H(\partial)-cK(\partial)$,
we have, as a consequence of Lemma \ref{20131027:prop1_mat},
two compatible Poisson structures $K$ and $H$ 
on $\mc V_{N,m}^{\infty}$ (resp. $\mc V_{N,m}$),
which we call, respectively, the
\emph{first} and the \emph{second matrix AGD Poisson structures}
(cf. Definition \ref{agd}).
Using equation \eqref{h_mat} we get the following explicit formulas 
for the $\lambda$-brackets associated to $H$ and $K$ ($a,b,c,d=1,\dots,m$):
\begin{equation}\label{eq:H_mat}
\begin{array}{c}
\vphantom{\Big(}
\{L_{ab}(z)_\lambda L_{cd}(w)\}_H
=L_{cb}(z)i_z(z-w-\lambda-\partial)^{-1}L_{ad}(w)
\\
\vphantom{\Big(}
-L_{cb}(w+\lambda+\partial)i_z(z-w-\lambda-\partial)^{-1}L_{ad}^*(-z+\lambda)
\end{array}
\end{equation}
and
\begin{equation}\label{eq:K_mat}
\begin{array}{c}
\displaystyle{
\{L_{ab}(z)_\lambda L_{cd}(w)\}_K
=\delta_{ad}i_z(z-w-\lambda)^{-1}\left(L_{cb}(z)-L_{cb}(w+\lambda)\right)
}\\
\displaystyle{
+\delta_{cb}i_z(z-w-\lambda-\partial)^{-1}\left(L_{ad}(w)-L_{ad}^*(\lambda-z)\right)\,.
}
\end{array}
\end{equation}
Expanding equations \eqref{eq:H} and \eqref{eq:K}
in powers of $z$ and $w$, we get ($i,j\in I$):
\begin{equation}\label{HKij_mat}
\begin{array}{l}
\displaystyle{
H_{ji;cdab}(\lambda)
=
\sum_{k,\alpha\in\mb Z_+}\binom{k}{\alpha}u_{i-k-1,cb}
(\lambda+\partial)^\alpha u_{j+k-\alpha,ad}
}\\
\displaystyle{
-\sum_{k,\alpha,\beta\in\mb Z_+}(-1)^\alpha\binom{j}{\alpha}
\binom{i-k-1}{\beta}u_{j+k-\alpha,cb}
(\lambda+\partial)^{\alpha+\beta}u_{i-\beta-k-1,ad}
\,,}
\\
\displaystyle{
K_{ji;cdab}(\lambda)
=
\epsilon_{ij}
\sum_{k\in\mb Z_+}\left(
\binom{i}{k}\delta_{cb}(\lambda+\partial)^ku_{i+j-k,ad}
-\binom{j}{k}\delta_{ad}(-\lambda)^ku_{i+j-k,cb}\right)
\,,}
\end{array}
\end{equation}
where, as in \eqref{HKij}, $\epsilon_{ij}=+1$ if $i,j\in\mb Z_+$, $\epsilon_{ij}=-1$ if $i,j<0$,
and $\epsilon_{ij}=0$ otherwise.
\begin{remark}\label{gcN}
As in Remark \ref{gc1},
the $\mb F[\partial]$-submodule $R_+\subset \mc V^\infty_{N,m}$
generated by $u_{i,ab}$, $i\in I,a,b\in\{1,\dots,m\}$
is closed with respect to the $K$-$\lambda$-bracket,
and it is a Lie conformal algebra isomorphic to $\mf{gc}_N$
(see \cite{Kac96}).
\end{remark}
\begin{example}\label{20130722:exa1}
For $N=1$, we have $\mc V_{1,m}=\mb F[u_{ab}^{(n)}\mid a,b=1,\dots,m,n\in\mb Z_+]$
and $L(\partial)=\partial\mbb1_m+U\in\mc M_{1,m}[\partial]$.
In this case we have
$$
\{u_{ab}{}_{\lambda}u_{cd}\}_{H}
=\delta_{bc}u_{ad}-\delta_{da}u_{cb}-\delta_{ad}\delta_{cb}\lambda\,,
$$
for any $a,b,c,d=1,\dots,m$ and $K=0$.
This is the affine PVA $S(\mb F[\partial]\mf{gl}_m)$
associated to the Lie algebra $\mf{gl}_m$ and its trace form.
\end{example}

\subsection{The classical matrix \texorpdfstring{$\mc W$}{W}-algebras
\texorpdfstring{$\mc W_{N,m}$}{W_Nm}}\label{sub:V-algebras}

\begin{lemma}\label{20130926:lem1_mat}
In the AGD bi-PVA $\mc V_{N,m}^{\infty}$ (resp. $\mc V_{N,m}$) we have, for all $a,b,c,d=1,\dots,m$:
\begin{enumerate}[(a)]
\item
$\{u_{-N,ab}{}_\lambda L_{cd}(w)\}_H
=\delta_{cb}L_{ad}(w)-\delta_{ad}L_{cb}(w+\lambda)$;
\item
$\{L_{ab}(z)_\lambda u_{-N,cd}\}_H
=\delta_{cb}L^*_{ad}(-z+\lambda)-\delta_{ad}L_{cb}(z)$;
\item
$\{u_{-N,ab}{}_{\lambda}u_{-N,cd}\}_H
=\delta_{cb}u_{-N,ad}-\delta_{ad}u_{-N,cb}-\delta_{ad}\delta_{cb}N\lambda$;
\item
$\{u_{-N,ab}{}_\lambda L_{cd}(w)\}_K=\{L_{ab}(z){}_\lambda u_{-N,cd}\}_K=0$.
\end{enumerate}
\end{lemma}
\begin{proof}
Same as the proof of Lemma \ref{20130926:lem1}.
\end{proof}
Let $C_{abcd}(\lambda)=\{u_{-N,cd}{}_{\lambda}u_{-N,ab}\}_H$, $a,b,c,d\in\{1,\dots,m\}$.
By Lemma \ref{20130926:lem1_mat}(c)
and Proposition \ref{prop:roots}(b), 
the corresponding matrix differential operator
\begin{equation}\label{20140107:eq1}
C(\partial)=(C_{abcd}(\partial))_{abcd=1}^m\in\Mat_{m^2\times m^2}\mc V_{N,m}
\,,
\end{equation}
is invertible. 
Furthermore, by \ref{20130926:lem1_mat}(d), the elements $u_{-N,ab}$ are central
with respect to the Poisson structure $K$, for all $a,b=1,\dots,m$.
Therefore,
by Theorems \ref{prop:dirac} and \ref{20130516:thm1} we can perform the Dirac reduction
to get a bi-Poisson structure $(H^D,K)$ on 
$$
\begin{array}{l}
\displaystyle{
\quot{\mc V_{N,m}^\infty}{\langle u_{-N,ab}\rangle_{a,b=1}^m}
\simeq\mb F\Big[u_{i,ab}^{(n)}\Big| \substack{-N\neq i\in I, n\in\mb Z_+\\ a,b\in\{1,\dots,m\}}\Big]
=:{\mc W}_{N,m}^\infty
} \\
\displaystyle{
\Big(\text{ resp. } 
\quot{\mc V_{N,m}}{\langle u_{-N,ab}\rangle_{a,b=1}^m}
\simeq\mb F\Big[u_{i,ab}^{(n)}\,\Big| \substack{-N\neq i\in I_-, n\in\mb Z_+\\ a,b\in\{1,\dots,m\}}\Big]
=:\mc W_{N,m}
\Big)\,.}
\end{array}
$$
We can compute the non-local $\lambda$-brackets corresponding to $H^D$
using Lemma \ref{20130926:lem1_mat}.
In terms of the generating series, 
we have ($a,b,c,d=1,\dots,m$)
\begin{equation}\label{eq:H_dirac_mat}
\begin{array}{l}
\vphantom{\Big(}
\displaystyle{
\{L_{ab}(z)_{\lambda}L_{cd}(w)\}_{H^D}
=
L_{cb}(z)i_z(z-w-\lambda-\partial)^{-1}L_{ad}(w)
} \\
\vphantom{\Big(}
\displaystyle{
-L_{cb}(w+\lambda+\partial)i_z(z-w-\lambda-\partial)^{-1}L_{ad}^*(-z+\lambda)
} \\
\vphantom{\Big(}
\displaystyle{
-\frac1NL_{cb}(w+\lambda+\partial)(\lambda+\partial)^{-1}L^*_{ad}(-z+\lambda)
-\frac1NL_{ad}(w)(\lambda+\partial)^{-1}L_{cb}(z)
} \\
\vphantom{\Big(}
\displaystyle{
+\frac{1}{N}
\sum_{k=1}^m
\delta_{ad}L_{ck}(w+\lambda+\partial)(\lambda+\partial)^{-1}L_{kb}(z)
} \\
\vphantom{\Big(}
\displaystyle{
+\frac{1}{N}
\sum_{k=1}^m
\delta_{cb}L_{kd}(w)(\lambda+\partial)^{-1}L^*_{ak}(-z+\lambda)
\,.}
\end{array}
\end{equation}

It is not hard to show that $T=\tr\res_{z}L(z)z^{-N+1}$
is a Virasoro element with central charge $\frac{m(N^3-N)}{12}$ 
(cf. Definition \ref{def:CFTtype}).
Moreover, we have, for all $a,b=1,\dots,m$
$$
\{T_{\lambda}u_{j,ab}(w)\}_{H^D}
=\big(\partial+(N+j+1)\lambda\big)u_{j,ab}+O(\lambda^2)\,,
$$
i.e. $u_{j,ab}$ is a $T$-eigenvector of conformal weight $N+j+1$.
It was proved by \cite{Bil95} that $\mc W_{N,m}$ has a 
differential basis given by $T$ and primary elements,
i.e. it is a non-local PVA of CFT type.
\begin{example}\label{exa:matrixkdv_poisson}
For $N=2$ we have $\mc W_{2,m}=\mb F[u_{ab}^{(n)}\mid a,b\in\{1,\dots,m\},n\in\mb Z_+]$.
The formula of the $\lambda$-brackets \eqref{eq:H_dirac_mat} and \eqref{eq:K_mat}
is ($a,b,c,d=1,\dots,m$):
\begin{equation}\label{matrixkdv_lambda}
\begin{array}{l}
\displaystyle{
\{u_{ab}{}_{\lambda}u_{cd}\}_{H^D}
=\delta_{ad}\delta_{cb}\frac{\lambda^3}{2}
+\delta_{cb}(\lambda+\frac{\partial}{2})u_{ad}
+\delta_{ad}(\lambda+\frac{\partial}{2})u_{cb}
-\frac12u_{ad}(\lambda+\partial)^{-1}u_{cb}
}\\
\displaystyle{
-\frac12u_{cb}(\lambda+\partial)^{-1}u_{ad}
+\frac12\sum_{k=1}^m\left(
\delta_{ad}u_{ck}(\lambda+\partial)^{-1}u_{kb}
+\delta_{cb}u_{kd}(\lambda+\partial)^{-1}u_{ak}
\right)\,,
}\\
\displaystyle{
\{u_{ab}{}_{\lambda}u_{cd}\}_{K}
=2\delta_{ad}\delta_{bc}\lambda\,.
}
\end{array}
\end{equation}
This is the same bi-Poisson structure
considered by \cite{OS98} when studying the non-commutative KdV equation.
For $m=2$, let
$T=u_{11}+u_{22}$, $v=u_{11}-u_{22}$, $v_+=u_{12}$ and $v_-=u_{21}$.
Then (cf. \cite{Bil95})
$$
\begin{array}{l}
\vphantom{\Big(}
\displaystyle{
\{T_\lambda T\}_c
=(2\lambda+\partial)T+\lambda^3-4c\lambda
\,,\,\,
\{T_\lambda v\}_c
=(2\lambda+\partial)v
\,,\,\,
\{T_\lambda v_{\pm}\}_c
=(2\lambda+\partial)v_{\pm}
\,,} \\
\vphantom{\Big(}
\displaystyle{
\{v_\lambda v\}_c
=2v_+(\lambda+\partial)^{-1}v_-+2v_-(\lambda+\partial)^{-1}v_+
+(2\lambda+\partial)T+\lambda^3-4c\lambda
\,,} \\
\vphantom{\Big(}
\displaystyle{
\{v_\lambda v_{\pm}\}_c
=-v(\lambda+\partial)^{-1}v_{\pm}
\,,\,\,
\{v_{\pm}{}_\lambda v_{\pm}\}_c
=-v_{\pm}(\lambda+\partial)^{-1}v_{\pm}
\,,} \\
\vphantom{\Big(}
\displaystyle{
\{v_+{}_\lambda v_-\}_c
=\frac12 v(\lambda+\partial)^{-1}v+v_-(\lambda+\partial)^{-1}v_+
+\frac12(2\lambda+\partial)T+\frac12\lambda^3-2c\lambda
\,.}
\end{array}
$$
Note that $T$ is a Virasoro element with central charge $1$ 
and $v,v_\pm$ are primary elements of conformal weight $2$.
Hence $\mc W_{2,2}$ is a non-local PVA of CFT type.
\end{example}

\begin{remark}
The following matrix analogue of the Kupershmidt-Wilson Theorem \ref{kw_thm} holds.
Let $R_{N,m}$ be the tensor product of $N$-copies 
of the affine Poisson vertex algebra for $\mf{gl}_m$:
$R_{N,m}=\mb F[v_{i,ab}^{(n)}\mid i\in\{1,\dots,N\},a,b\in\{1,\dots,m\},n\in\mb Z_+]$,
with the $\lambda$-bracket
($i,j\in\{1,\dots,N\}$, $a,b,c,d\in\{1,\dots,m\}$)
$$
\{ {v_{i,ab}}_\lambda {v_{j,cd}} \} 
= \delta_{ij}\big(\delta_{bc}v_{ad}-\delta_{ad}v_{cb}+\delta_{ad}\delta_{cb}\lambda\big)
\,.
$$
We have a PVA homomorphism
$\mu:\,\mc V_{N,m}\to R_{N,m}$, 
given by
\begin{align}\label{miuraop_mat}
\mu(L(\partial))=(\partial+V_N)(\partial+V_{N-1})\cdots(\partial+V_1)
\in \Mat_{m\times m}R_{N,m}[\partial]\,,
\end{align}
where $L(\partial)\in\mc M_{N,m}[\partial]$ 
is as in \eqref{lcap_mat},
and $V_i=\big(v_{i,ab}\big)_{a,b=1}^m$.
Moreover, we can consider the Dirac reduction of $R_{N,m}$
by the constraints $\theta_{ab}=\sum_{k=1}^N v_{k,ab}$, for $a,b\in\{1,\dots,m\}$.
Then the map \eqref{miuraop_mat}
induces a PVA homomorphism on the Dirac reductions,
called the \emph{matrix Miura map}:
$$
\mu:\,\mc W_{N,m}\to\quot{R_{N,m}}{\langle \theta_{ab}\mid a,b=1,\dots,m\rangle}\,.
$$
\end{remark}

\subsection{Gelfand-Dickey Integrable hierarchies in the matrix case}\label{sec:hier_mat}

Throughout this section we let $\mc V=\mc V^\infty_{N,m}$ or $\mc V_{N,m}$,
endowed with its AGD bi-Poisson structure $(H,K)$,
and we let $\mc M=\Mat_{m\times m}\mc V$.
Let $L(\partial)$ as in \eqref{lcap_mat}.
In analogy with Theorem \ref{prop:lenard_works},
we shall prove that the sequence of local functionals $\{\tint h_k\}_{k\geq1}$,
where
$$
h_k=\frac Nk \tr \res_z L^{\frac kN}(z)\in\mc V
\,,
$$
satisfies the Lenard-Magri recursive condition \eqref{lenardrecursion},
and it spans an infinite dimensional subspace of $\quot{\mc V}{\partial\mc V}$.
Hence,
we get the corresponding integrable hierarchy of bi-Hamiltonian equations
$$
\frac{du_{i,ab}}{dt_k}
=
\{\tint h_k,u_{i,ab}\}_H
=
\{\tint h_{k+N},u_{i,ab}\}_K
\,.
$$
This is a consequence of the following results,
which are proved in the same way as in Section \ref{sub:hierarchies1},
for the scalar case.

\begin{lemma}\label{lem:15032013_mat}
Let $\mc V$ be a differential algebra endowed with a $\lambda$-bracket
$\{\cdot\,_{\lambda}\,\cdot\}$,
and let $\mc M=\Mat_{m\times m}\mc V$.
Let $L(\partial)\in\mc M((\partial^{-1}))$ be a monic matrix pseudodifferential operator
of order $N>0$.
Then, for all $k\geq1$ and $a,b\in\{1,\dots,m\}$, 
the following identity holds in $\mc V((w^{-1}))$:
\begin{equation}\label{eq:lenard1_mat}
\begin{array}{l}
\displaystyle{
\res_z
\{\tr L^{\frac kN}(z)_{\lambda}L_{ab}(w)\}\big|_{\lambda=0}
} \\
\displaystyle{
=\frac kN
\sum_{c,d=1}^m
\res_z\{L_{cd}(z+x)_x L_{ab}(w)\} \big(\big|_{x=\partial}(L^{\frac kN-1})_{dc}(z)\big)
\,.}
\end{array}
\end{equation}
\end{lemma}
\begin{lemma}\label{prop:var_der_mat}
For every $i\in I$ (resp. $I_-$), $a,b=1,\dots,m$, and $k\geq1$, 
we have, in $\mc V=\mc V_{N,m}^\infty$ (resp. $\mc V_{N,m}$),
\begin{equation}\label{eq:varder_mat}
\frac{\delta h_k}{\delta u_{i,ab}}
=\res_z(z+\partial)^{-i-1}(L^{\frac kN-1})_{ba}(z)\,.
\end{equation}
\end{lemma}
\begin{lemma}\label{20131107:lem_mat}
The local functionals $\tint h_k\in\quot{\mc V}{\partial\mc V},\, k\in\mb Z_+\backslash N\mb Z_+$, 
are all linearly independent.
\end{lemma}
\begin{lemma}\label{cor:17032013_mat}
For the AGD bi-PVA $\mc V=\mc V_{N,m}^\infty$ or $\mc V_{N,m}$,
with bi-Poisson structure $(H,K)$, we have ($a,b=1,\dots,m$, $k\geq1$)
\begin{enumerate}[(a)]
\item
$\displaystyle{
\{h_k{}_\lambda L_{ab}(w)\}_H\big|_{\lambda=0}
=\sum_{c=1}^m \Big(
(L^{\frac kN})_{ac}(w+\partial)_+L_{cb}(w)-L_{ac}(w+\partial)(L^{\frac kN})_{cb}(w)_+
\Big)}$;
\item
$\displaystyle{
\{h_k{}_\lambda L_{ab}\!(w)\}_K\big|_{\lambda=0}
\!\!
=\!\!\sum_{c=1}^m\! \Big(\!
(\!L^{\frac{k}{N}-1}\!)_{ac}(w+\partial)_+L_{cb}(w)-L_{ac}(w+\partial)(\!L^{\frac{k}{N}-1}\!)_{cb}(w)_+\!
\Big)}$.
\end{enumerate}
\end{lemma}
\begin{lemma}\label{lem:lenard_works_mat}
For every $k\geq1$, 
we have the Lenard-Magri recursion
\begin{equation}\label{leneq_mat}
\{h_k{}_\lambda u\}_H\big|_{\lambda=0}
=\left.\{h_{k+N}{}_\lambda u\}_K\right|_{\lambda=0}
\,,\,\,
\text{ for all } u\in\mc V
\,.
\end{equation}
\end{lemma}
\begin{lemma}\label{lem:lenard_start_mat}
For every $\varepsilon\in\{1,\dots,N\}$, 
we have 
\begin{equation}\label{leneq2_mat}
\{{h_{\varepsilon}}_\lambda u\}_K\big|_{\lambda=0}=0
\,,\,\,
\text{ for all } u\in\mc V
\,.
\end{equation}
\end{lemma}
\begin{remark}\label{20132507:rem2}
As already noted in Remark \ref{20132507:rem1} in the scalar case,
it follows from Lemma \ref{cor:17032013_mat} and \eqref{eq:mult_symbol}
that the Hamiltonian equation corresponding to the Hamiltonian functional
$\tint h_k$, $k\geq1$, can be written as
\begin{equation}\label{eq:lax_pres}
\frac{dL(w)}{dt_k}=[(L^{\frac kN})_+,L](w)\,,
\end{equation}
in the space $\mc M((w^{-1}))$.
\end{remark}

\subsection{Integrable hierarchies for the reduction to the case \texorpdfstring{$U_{-N}=0$}{U_-N=0}}

Let, as before, $\mc V=\mc V_{N,m}$ and let $\mc K$ be the (differential) field of fractions of $\mc V$.
We shall identify the space $\Mat_{m\times m}\mc V$ of $m\times m$ matrices with coefficients in $\mc V$
with the space $\mc V^{m^2}$ in the obvious way.
Let, as in \eqref{lcap_mat}, 
\begin{equation}\label{lcap_mat2}
L(\partial)
=\partial^N\mbb1_m+U_{-N}\partial^{N-1}+\ldots+U_{-1}
\in\Mat_{m\times m}\mc V[\partial]
\,,
\end{equation}
where $U_{i}=\big(u_{i,ab}\big)_{a,b=1}^m$
and the $u_{i,ab}$'s are the generators of the algebra of differential polynomials $\mc V$.

Recall from Section \ref{sec:2.2b_mat}
that we have a bi-Poisson structure $H,K\in\Mat_{Nm^2\times Nm^2}\mc V[\partial]$ 
on $\mc V$, given by \eqref{HKij_mat}
with indices running in $I_-$ and $u_{i,ab}=0$ for $i\geq0$.
Consider the the Dirac modification $H^D$ of $H$ by the constraints $u_{-N,ab},\,a,b=1,\dots,m$,
given by \eqref{eq:H_dirac_mat}.
Note that $H^D$ is a non-local Poisson structure,
hence the analogue of the Lenard-Magri recursion condition \eqref{leneq_mat}
has to be expressed in terms of association relations:
\begin{equation}\label{eq:lm2}
\frac{\delta h_{k}}{\delta u}\ass{H^D} P_k
\ass{K}\frac{\delta h_{k+N}}{\delta u}\,,
\end{equation}
for some $P_k\in\mc V^{Nm^2}$.
Recall from \cite{DSK13} that, for $X\in\Mat_{Nm^2\times Nm^2}\mc V((\partial^{-1}))$
and $F,P\in\mc V^{Nm^2}$,
the association relation $F\ass{X}P$ means that 
$X$ admits a fractional decomposition of the form $X=AB^{-1}$, 
where $A$ and $B$ are matrix differential operators with coefficients in $\mc V$ 
and $B$ is nondegenerate, and there exists $F_1\in\mc K^{Nm^2}$ such that $F=BF_1$ and $P=AF_1$.
In this section we shall prove the following (cf.  \cite[Ansatz 7.3]{DSKV13c}):
\begin{theorem}\label{20140114:thm}
The Lenard-Magri recursive conditions \eqref{eq:lm2} hold.
As a consequence, we get the induced integrable hierarchy of bi-Hamiltonian equations on $\mc W_{N,m}$:
$\frac{du_{i,ab}}{dt_k}=(P_k)_{i,ab}$, $i\in\{-N+1,\dots,-1\}$, $a,b\in\{1,\dots,m\}$, $k\geq1$.
\end{theorem}
In order to prove Theorem \ref{20140114:thm}
we apply the theory of singular degree of rational matrix pseudodifferential operators,
developed in \cite{CDSK13b}.
Recall that we can write (see \cite{DSKV13c})
\begin{equation}\label{21040114:eq1}
H^D=H+BC^{-1}B^*\,,
\end{equation}
where $B\in\Mat_{Nm^2\times m^2}\mc V[\partial]$ has entries
$B_{i,ab;cd}(\partial)=\res_w\{{u_{-N,cd}}_\partial L_{ab}(w)\}_Hw^i$,
and $C\in\Mat_{m^2\times m^2}\mc V[\partial]$ has entries
$C_{ab;cd}(\partial)=\{{u_{-N,cd}}_\partial u_{-N,ab}\}_H$.
Let $B(\lambda;\partial)_{ab,cd}=\sum_{i=-N}^{-1}B_{i,ab;cd}(\partial)\lambda^{-i-1}$.
By Lemma \ref{20130926:lem1_mat}(a) and (c), we get explicit formulas for
the matrices $B(\lambda;\partial)$ and $C(\partial)$,
considered as differential operators on $\mc V^{m^2}$:
\begin{equation}\label{21040114:eq2}
B(\lambda;\partial)F
=
[F^t,L](\lambda)
\,\,,\,\,\,\,
C(\partial)F
=
[F^t,U_{-N}]-N\partial F^t
\,,
\end{equation}
where $[F^t,L](\lambda)$ is the symbol of the differential operator
$[F^t,L(\partial)]=F^tL(\partial)-L(\partial)\circ F^t$.

The proof of Theorem \ref{20140114:thm} will be based on the following 7 Lemmas.
\begin{lemma}\label{20140114:lem1}
Let $N>1$.
The only solutions $X\in\Mat_{m\times m}\widetilde{\mc K}$,
where $\widetilde{\mc K}$ is any differential field extension of $\mc K$,
of the system of differential equations $[L(\partial),X]=0$
are constant scalar multiples of the identity matrix.
\end{lemma}
\begin{proof}
Equating to zero the coefficients of $\partial^{N-1}$ and $\partial^{N-2}$ of $[L(\partial),X]$,
we get the following system of differential equations:
\begin{equation}\label{20140114:eq2}
\begin{array}{l}
\vphantom{\Big)}
\displaystyle{
X^\prime=\frac1N[X,U_{-N}]
\,,}\\
\vphantom{\Big)}
\displaystyle{
\frac{N(N-1)}{2}X^{\prime\prime}+(N-1)U_{-N}X^\prime+[U_{-N+1},X]=0
\,.}
\end{array}
\end{equation}
Using the first equation, the second equation of \eqref{20140114:eq2} becomes
\begin{equation}\label{20140114:eq3}
\Big[X,
\frac{N-1}{2}U_{-N}^\prime+\frac{N-1}{2N}U_{-N}^2-U_{-N+1}
\Big]=0\,.
\end{equation}
In some differential field extension $\widetilde{\mc K}$ of $\mc K$,
there exists a non-degenerate matrix $E\in\Mat_{m\times m}\widetilde{\mc K}$
such that 
\begin{equation}\label{20140114:eq5}
E^\prime=\frac1N EU_{-N}\,.
\end{equation}
(The rows of $E$ form a basis of the space of solutions of the linear system $Y^\prime=\frac1N YU_{-N}$,
for $Y\in\widetilde{\mc K}^{m}$.)
It is immediate to see that, in this case,
\begin{equation}\label{20140114:eq6}
(E^{-1})^\prime=-\frac1N U_{-N}E^{-1}\,.
\end{equation}
Then, all solutions $X\in\Mat_{m\times m}\widetilde{\mc K}$
of the first equation in \eqref{20140114:eq2} are of the form
\begin{equation}\label{20140114:eq4}
X=E^{-1}CE\,,
\end{equation}
where $C$ is an arbitrary matrix with constant entries: $\partial C=0$.
(It is immediate to check, using \eqref{20140114:eq5} and \eqref{20140114:eq6},
that all the matrices $X$ as in \eqref{20140114:eq4}
solve the first equation in \eqref{20140114:eq2}.
On the other hand, the space of solutions of the first equation in \eqref{20140114:eq2}
is a vector space over the field of constants of dimension less than or equal to $m^2$, \cite{DSK13b}.)
Let then $V_{-N}=EU_{-N}E^{-1}$, and $V_{-N+1}=EU_{-N+1}E^{-1}$.
We immediately have by \eqref{20140114:eq5} and \eqref{20140114:eq6}
that $(V_{-N})^\prime=EU_{-N}^\prime E^{-1}$.
Hence, after conjugating by $E$,
equation \eqref{20140114:eq3} becomes
\begin{equation}\label{20140114:eq7}
\Big[C,
\frac{N-1}{2}V_{-N}^\prime+\frac{N-1}{2N}V_{-N}^2-V_{-N+1}
\Big]=0\,.
\end{equation}
Let $\widetilde{\mc C}\subset\widetilde{\mc K}$ be the subfield of constants in $\widetilde{\mc K}$.
Consider the map
$\widetilde{\mc C}[V_{-N,ab},V_{-N,ab}^\prime,V_{-N+1,ab}\,|\,a,b\in\{1,\dots,m\}]\to\widetilde{\mc K}$
given by conjugation by $E$.
It is obviously injective.
Hence, we can view \eqref{20140114:eq7}
as a system of equations in the polynomial algebra 
$\widetilde{\mc C}[V_{-N,ab},V_{-N,ab}^\prime,V_{-N+1,ab}\,|\,a,b\in\{1,\dots,m\}]$.
The coefficient of $V_{-N+1,ab}$ in the LHS of \eqref{20140114:eq7}
is $[E_{ab},C]$.
Therefore, equation \eqref{20140114:eq7} implies that $C=c\mbb1$,
for some $c\in\widetilde{\mc C}$.
The claim follows.
\end{proof}
\begin{lemma}\label{20140114:lem2}
The kernel of the operator $B(\lambda;\partial)$ on $\widetilde{\mc K}^{m^2}$
consists of constant scalar multiples of the identity matrix,
and it is contained in the kernel of $C(\partial)$.
\end{lemma}
\begin{proof}
The first claim is obvious by the definition \eqref{20140114:eq1} of $B(\lambda;\partial)$
and Lemma \ref{20140114:lem1}.
For the second claim we just need to observe that $C(\partial)$
is the coefficient of $\lambda^{N-1}$ in $B(\lambda;\partial)$.
\end{proof}
\begin{lemma}\label{20140114:lem3}
Let $Q(\partial):\,\mc V^{m^2}\to\mc V^{m^2}$ be the following differential operator
\begin{equation}\label{20140114:eq1}
Q(\partial)F=F-F_{mm}\mbb1+\partial F_{mm}E_{mm}\,.
\end{equation}
Then $B(\lambda;\partial)$ and $C(\partial)$ are both divisible on the right by $Q(\partial)$,
i.e. $B(\lambda;\partial)=B_1(\lambda;\partial)\circ Q(\partial)$,
$C(\partial)=C_1(\partial)\circ Q(\partial)$.
Moreover, the kernels of $B_1(\lambda;\partial)$ and $C_1(\partial)$
have zero intersection in $\widetilde{\mc K}^{m^2}$
for any differential field extension $\widetilde{\mc K}$ of $\mc K$.
\end{lemma}
\begin{proof}
The operator $Q(\partial)$
is a non-degenerate matrix differential operator of degree 1,
and its kernel is $\mb F\mbb1$.
The claim follows by \cite[Thm.4.4]{CDSK13a}.
\end{proof}
\begin{lemma}\label{20140114:lem4}
We have $C_1^*(\partial)\mbb1\neq0$,
where $C_1(\partial):\,\mc V^{m^2}\to\mc V^{m^2}$ is as in Lemma \ref{20140114:lem3},
and $C_1^*(\partial)$ is its formal adjoint.
\end{lemma}
\begin{proof}
It is straightforward to check that $C(\partial)=C_1(\partial)\circ Q(\partial)$,
where $C_1(\partial)$ is the following matrix differential operator
$$
C_1(\partial)F
=
[F^t,U_{-N}]-N\partial F^t
-F_{mm}(N\mbb1+[E_{mm},U_{-N}])+NF_{mm}^\prime E_{mm}\,.
$$
Its formal adjoint is
$$
C_1^*(\partial)F
=
[U_{-N},F^t]+N\partial F^t
-\tr\big(F^t(N\mbb1+[E_{mm},U_{-N}])\big)E_{mm}
-NF_{mm}^\prime E_{mm}\,.
$$
Hence, $C_1^*(\partial)\mbb1=-NmE_{mm}\neq0$.
\end{proof}
\begin{lemma}\label{20140114:lem5}
The kernels of the operators $B(\lambda,\partial)$ and $C_1^*(\partial)$
have zero intersection in $\widetilde{\mc K}^{m^2}$,
for any differential field extension $\widetilde{\mc K}$ of $\mc K$.
\end{lemma}
\begin{proof}
It follows from Lemmas \ref{20140114:lem2} and \ref{20140114:lem4}.
\end{proof}
\begin{lemma}\label{20140114:lem6}
$H^D=H+B_1C_1^{-1}B^*$ is a minimal rational expression for $H^D$.
\end{lemma}
\begin{proof}
It follows from \cite[Cor.4.11]{CDSK13b}.
Indeed, 
the space $\mc E$ defined in \cite[Eq.(4.30)]{CDSK13b} is, in this case,
$\mc E=\ker(B_1)\cap\ker(C_1)\subset\widetilde{\mc K}^{m^2}$, 
which is zero by Lemma \ref{20140114:lem3},
and the space $\mc E^*$ defined in \cite[Eq.(4.31)]{CDSK13b} is, in this case,
$\mc E^*=\ker(B)\cap\ker(C_1^*)\subset\widetilde{\mc K}^{m^2}$, 
which is zero by Lemma \ref{20140114:lem5}.
\end{proof}
\begin{lemma}\label{20140114:lem7}
$B^*(\partial)\frac{\delta h_k}{\delta u}=0$ for every $k\geq1$.
\end{lemma}
\begin{proof}
The symbols of the matrix elements of $B^*$ are
($i\in\{-N,\dots,-1\}$, $a,b,c,d\in\{1,\dots,m\}$)
$$
B^*_{ab;i,cd}(\lambda)
=\res_z\big(
\delta_{ad}L_{cb}^*(-z+\lambda)-\delta_{cb}L_{ad}(z)
\big)w^i\,.
$$
Using Lemma \ref{prop:var_der_mat} we have
$$
\begin{array}{l}
\displaystyle{
\big(B^*(\partial)\frac{\delta h_k}{\delta u_i})_{ab}
=\sum_{c=1}^m\res_z\res_w
L_{cb}^*(-z+\partial)
\delta(z-w-\partial)
(L^{\frac kN-1})_{ac}(w)
}
\\
\displaystyle{
-\sum_{d=1}^m\res_z\res_w
L_{ad}(z)
\delta(z-w-\partial)
(L^{\frac kN-1})_{db}(w)
\,.}
\end{array}
$$
Using equation \eqref{deltaprop} and Lemma \ref{lem:residui}(b) we can rewrite
$$
\begin{array}{l}
\displaystyle{
\sum_{c=1}^m\res_z\res_w
L_{cb}^*(-z+\partial)
\delta(z-w-\partial)
(L^{\frac kN-1})_{ac}(w)
}\\
\displaystyle{
=\sum_{c=1}^m
\res_w L_{cb}^*(-w)
(L^{\frac kN-1})_{ac}(w)
=\res_w(L^{\frac kN})_{ab}(w)
}
\end{array}
$$
and
$$
\begin{array}{l}
\displaystyle{
\sum_{d=1}^m\res_z\res_w
L_{ad}(z)
\delta(z-w-\partial)
(L^{\frac kN-1})_{db}(w)
}\\
\displaystyle{
=
\sum_{d=1}^m\res_w
L_{ad}(w+\partial)
(L^{\frac kN-1})_{db}(w)
=\res_w(L^{\frac kN})_{ab}(w)\,.
}
\end{array}
$$
This shows that $\big(B^*(\partial)\frac{\delta h_k}{\delta u_i})_{ab}=0$
for all $a,b=1,\dots,m$ and $k\geq1$.
Note that, in fact, this Lemma follows from a general result on Dirac reduction,
see \cite[Lem.5.2(b)]{DSKV13c}.
\end{proof}
\begin{proof}[{Proof of Theorem \ref{20140114:thm}}]
Condition $\frac{\delta_{h_k}}{\delta u}\ass{K} P_k$ holds by assumption.
By Lemma \ref{20140114:lem6} and \cite[Theorem 4.12]{CDSK13b},
condition $\frac{\delta_{h_k}}{\delta u}\ass{H^D} P_k$
is equivalent to $B^*(\partial)\frac{\delta h_k}{\delta u}=0$,
which holds by Lemma \ref{20140114:lem7}.
As in the scalar case, it is not dificult to show that the elements
$\frac{\delta h_k}{\delta u},\,k\in\mb Z_+$, are linearly independent.
The claim follows.
\end{proof}

\begin{example}[$\mc W_{2,m}$: the Korteweg-de Vries hierarchy]
Recall from Example \ref{exa:matrixkdv_poisson} that
$\mc W_{2,m}=\mb F[u_{ab}^{(n)}\mid a,b\in\{1,\dots,m\},n\in\mb Z_+]$,
with the bi-PVA structure as in \eqref{matrixkdv_lambda}.
The first few fractional powers of 
$L(\partial)=\partial^2+U$, where $U=(u_{ab})_{a,b=1}^m$, are
$$
\begin{array}{l}
\displaystyle{
L^{\frac12}(\partial)
=\partial+\frac12U\partial^{-1}
-\frac14U'\partial^{-2}+\frac18(U''-U^2)\partial^{-3}+\dots
\,,} \\
\displaystyle{
L^{\frac32}(\partial)
=\partial^3+\frac32U\partial+\frac34U'
+\frac18(3U^2-U'')\partial^{-1}+\dots
\,,}
\end{array}
$$
from which we get 
$L^{\frac12}(w)_+=w$, $L^{\frac32}(w)_+=w^3+\frac{3}{4}(2w+\partial)U$,
and $\tint h_1=\tint \tr U$, $\tint h_3=\tint\frac14\tr U^2$.
The corresponding Hamiltonian equations \eqref{eq:hierarchy} are
$\frac{dU}{dt_1}=U'$ and the \emph{matrix Korteweg-de Vries equation}
\begin{equation}\label{eq:matrix_kdv}
\frac{dU}{dt_3}=\frac14(U'''+3UU'+3U'U)\,.
\end{equation}
Thus we proved, in particular, that this equation is integrable,
as was conjectured by \cite{OS98}.
For $m=2$ (see Example \ref{exa:matrixkdv_poisson} for notation),
the Hamiltonian functionals, expressed in terms of the matrix elements of $U$,
become $\tint h_0=\tint T$ and $\tint h_1=\frac12\tint (T^2+2v_+v_-+v^2)$,
and the matrix equation \eqref{eq:matrix_kdv} becomes the system
$$
\left\{\begin{array}{l}
\displaystyle{
\frac{dT}{dt_3}=\frac14(T'''+3(T^2+v_+v_-+v^2)')\,,
}\\
\displaystyle{
\frac{dv}{dt_3}=\frac14(v'''+6(vT)')\,,
}\\
\displaystyle{
\frac{dv_{\pm}}{dt_3}=\frac14(v_{\pm}'''+6(v_{\pm}T)')\,.
}
\end{array}\right.
$$
\end{example}

\begin{remark}\label{20140107:rem}
The results of \cite{CDSK13b} only apply to the case of matrices of finite size.
Hence, the proof of Theorem \ref{20140114:thm}
does not work for the bi-Poisson structure $(H^D,K)$ on $\mc W_{N,m}^\infty$.
However, we expect that the claim of Theorem \ref{20140114:thm}
holds also in this case.
\end{remark}

\begin{example}[$\mc W^\infty_{1,m}$: the matrix KP hierarchy]
On $\mc W^\infty_{1,m}=\mb F[u_{i,ab}^{(n)}\mid i,n\in\mb Z_+,a,b\in\{1,\dots,m\}]$,
we have
$L(\partial)=\partial+\sum_{i\in\mb Z_+}U_i\partial^{-i-1}$,
where $U_i=(u_{i,ab})_{a,b=1}^m$.
The first few integrals of motion $\tint h_k$, $k\geq1$, are
$$
\begin{array}{l}
\displaystyle{
\tint h_1=\tint \tr U_0=\int\sum_{a=1}^mu_{aa;0}\,,
\qquad
\tint h_2=\tint \tr U_1=\int\sum_{a=1}^mu_{aa;1}\,,
}
\\
\displaystyle{
\tint h_3=\tint \tr(U_2+U_0^2)=\int\sum_{a=1}^mu_{aa;2}
+\sum_{a,b=1}^mu_{ab;0}u_{ba;0}\,,
\dots
}
\end{array}
$$
To find the corresponding bi-Hamiltonian equations,
we use Lemma \ref{cor:17032013_mat}.
We have
$L(w)_+=w$, $L^2(w)_+=w^2+2U_0$ and $L^3(w)_+=w^3+3U_0w+3(U_1+U_0')$. 
Hence, 
\begin{equation}\label{KP_mat}
\begin{array}{l}
\displaystyle{
\frac{dL(w)}{dt_1}=\partial L(w)\,,
\quad
\frac{dL(w)}{dt_2}
=\partial^2L(w)+2w\partial L(w)+2U_0L(w)-2L(w+\partial)U_0\,,
}
\\
\displaystyle{
\frac{dL(w)}{dt_3}=
\partial^3L(w)+3w\partial^2L(w)+3w^2\partial L(w)+3U_0\partial L(w)
}\\
\displaystyle{
+3w\left(U_0L(w)-L(w+\partial)U_0\right)
+3\left((U_1+U_0')L(w)-L(w+\partial)(U_1+U_0')\right)\,.
}
\end{array}
\end{equation}
As done in Remark \ref{rem:kp},
we can eliminate the variable $U_2$ from
the first two equations in the second system of the hierarchy \eqref{KP_mat},
and the first equation in the third system of \eqref{KP_mat}.
After relabeling $t_1=y$, $t_2=t$, $U=U_0$ and $W=2U_1+U_0'$, we get
the system
\begin{equation}\label{20140107:eq2}
\left\{
\begin{array}{l}
\displaystyle{
W'=U_y\,,
}\\
\displaystyle{
3W_y=4U_t-U'''-6(U^2)'+6[U,W]\,,
}
\end{array}\right.
\end{equation}
which we call the \emph{matrix Kadomtsev-Petviashvili equation} (when $m=1$
it reduces to the usual Kadomtsev-Petviashvili equation \eqref{eq:kp}).
According to Remark \ref{20140107:rem},
we expect that equation \eqref{20140107:eq2}
belongs to an infinite hierarchy of integrable bi-Hamiltonian equations.
\end{example}

\acks
We wish to thank Vladimir Sokolov and Chao-Zhong Wu for useful discussions.
We are also grateful to IHES, University of Rome and MIT for their kind hospitality.

The first author is supported by the national FIRB grant RBFR12RA9W.
%
%
The third author is supported by the ERC grant
``FroM-PDE: Frobenius Manifolds and Hamiltonian Partial
Differential Equations''.

\end{document}